\documentclass{article}
\usepackage{amsfonts,amssymb,mathrsfs}

\usepackage{graphicx,color}
\usepackage{amsmath,epstopdf}

\newtheorem{theorem}{Theorem}

\newtheorem{definition}[theorem]{Definition}

\newtheorem{lemma}[theorem]{Lemma}
\newtheorem{demonstration}[theorem]{Demonstration}

\newtheorem{proposition}[theorem]{Proposition}

\newenvironment{proof}[1][Proof]{\textbf{#1.} }{\ \rule{0.5em}{0.5em}}

\definecolor{green}{rgb}{0.00,0.50,0.00}

\begin{document}

\title{A Quantum Kalman Filter-Based PID Controller}
\author{John E.~Gough \footnote{jug@ber.ac.uk}\\
Aberystwyth University, SY23 3BZ, Wales, United Kingdom}
\date{}
\maketitle

\begin{abstract}
 We give a concrete description of a controlled quantum stochastic dynamical model corresponding to a quantum system (a cavity mode) under going continual quadrature measurements, with a PID controller acting on the filtered estimate for the mode operator. Central use is made of the input and output pictures when constructing the model: these unitarily equivalent pictures are presented in the paper, and used to transfer concepts relating to the controlled internal dynamics to those relating to measurement output, and \text{vice versa}.
The approach shows the general principle for investigating mathematically and physically consistent models in
which standard control theoretic methods are to be extended to the quantum setting.
\end{abstract}

\section{Introduction}\label{sec:introduction}
There has been considerable interest in the last decade in applying measurement-based feedback to quantum systems, especially after recognition of the Laboratoire Kastler Brossel (LKB) photon box experiments
resulting in the 2012 joint Nobel Prize for Physics to Serge Haroche \cite{Paris_PB,Rouchon}. It is reasonable to believe that classical feedback concepts will find their way into quantum systems if quantum technologies are ever to be realized. In this paper we will focus on PID controller implementations. There are a number of motivations for this. First, PID controllers are widely used in practice and so a large body of knowledge exists about the underlying theory and about their applications. Secondly, they would be simple devices to include in design problems as they have only three gain parameters to adjust. A third reason, and to a large extent the one pursued in this paper, is that we like to be able to export standard control techniques to the quantum domain and, to this end, need to know that the underlying modelling framework is both mathematically and physically correct. 

In the quantum domain, we have to worry that our dynamical model does not violate physical requirements such as the Heisenberg uncertainty relations, complete positivity, etc. If we wish to take the measured output of a quantum system and perform classical manipulations such a information processing and feedback, we need to be concerned that we have a hybrid classical-quantum description. In the quantum world, information cannot be copy and so the jump off points in classical block diagrams cannot be applied to quantum information. In this paper, we will give a description of a quantum PID controller which feeds back (classical) information arising from measurement - and to this end we need a quantum Kalman filer. Just as valuable as the specific structure of the model, is the fact that it qualifies as a consistent controlled quantum dynamical system. In principle, this shows the way to more detailed model building and opens up the programme of realizing standard feedback techniques in the quantum domain.

The approach adopted here relies on the notion of a controlled quantum stochastic evolution described by Luc Bouten and Ramon van Handel \cite{BvH_ref,BvH}, see also \cite{vH_thesis}, and brings to the fore the distinction between the input and the output pictures which is already implicit in their work.

The layout of the paper is as follows: in Section \ref{sec:filter} we recall the basic concepts of classical and quantum filtering. In Section \ref{sec:IO_pictures}, we develop the concept of input and output pictures, and emphasize the importance of working between these pictures in order to build up a consistent model. Finally, in
Section \ref{sec:PID_controller} we show how to construct the PI controllers as instances of feedback to the Hamiltonian. Here we have mechanisms to reduce two distinct notions of error: the Kalman filter seeks to reduce the error between our estimate and the unknown state while the controller aims to minimize the error between the
estimated state and a given reference signal. The D component of the controller is more problematic as typically it will amplify noise, which is an issue here as the filtered state is stochastic. Normally one smooths out the signal using a low-pass filter and this effectively restores a PI controller model, however, we treat the general PID problem here for completeness. The situation was addressed in the pioneering work of Howard Wiseman \cite{Wiseman} and in our classification this equates with a D filter component: the photo-current $y$ fed back in in \cite{Wiseman} being the singular process obtained as derivative of the measured process $Y$. The treatment of this involves a modification of the coupling operator coefficients of the system to the input noise fields in addition to a Hamiltonian correction. We build on the description in \cite{G_JMP_2017} of a PID filter to construct the appropriate model for a PID controller.

\subsection{Classical Feedback}
In practical applications one may take the control to be some function of the state,
or, in the case where that is not accessible, an estimate for the state. The latter will involve a filter,
and we will describe this in the next section, so let us look at the former problem first. For simplicity, we will assume that our system of interest - \textit{the plant} in engineering parlance - and that
the relationship between the control input $u$ and the (partial) observations $y$ of the state is linear with transfer function $\mathsf{G}$. 
In Figure \ref{fig:feedback_block} we have a
reference signal $r$ input at the start. From this we subtract the observation $y$ and the error $e=r-y$ is passed through a \textit{controller} which we assume to
be another linear system with transfer function $\mathsf{K}$. We have used double lines to emphasis that all the signals are classical and so can in principle be copied,
combined, and manipulated without restriction. \textit{ \textbf{N.B.} We use double-lined arrows to denote classical signals in feedback block diagrams. Later we will use single line arrows for quantum (i.e., operator-valued) signals.}

\begin{figure}[htbp]
	\centering
		\includegraphics[width=0.750\textwidth]{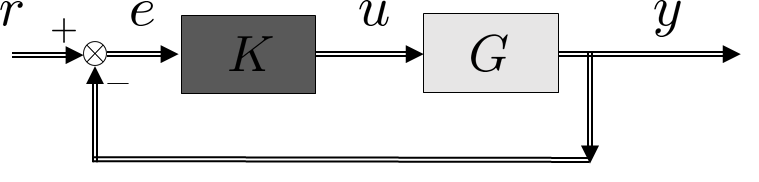}
	\caption{A plant $\mathsf{G}$ in a feedback loop with a controller $\mathsf{K}$.}
	\label{fig:feedback_block}
\end{figure}

We have $y=\mathsf{G} \, u \equiv \mathsf{GK} \,(r-y)$ from which we obtain
\begin{eqnarray}
y= \frac{\mathsf{GK}}{1+ \mathsf{GK}} \,  r .
\end{eqnarray}

The closed loop transfer function is $\mathsf{H} =\frac{\mathsf{GK}}{1+\mathsf{GK}}$ and we can consider design problems centering around our choice of the controller $\mathsf{K}$.

\subsection{PID Controllers}
By far the most widely used control algorithm in classical feedback applications are the PID controllers - with PI action being the most common form. For this class, we have the input-output relations
\begin{eqnarray}
u(t) = k_P e(t) + k_I \int_0^t e(t')dt' +
	k_D \dot{e} (t),
\end{eqnarray}
where the constants $k_P,k_I, k_D$ are the proportional, integral and derivative gains, respectively.
The associated transfer function is $\mathsf{K} (s) = k_P +s^{-1} k_I + s k_D$, see Figure \ref{fig:PID_filter}.

\begin{figure}[htbp]
	\centering
		\includegraphics[width=0.60\textwidth]{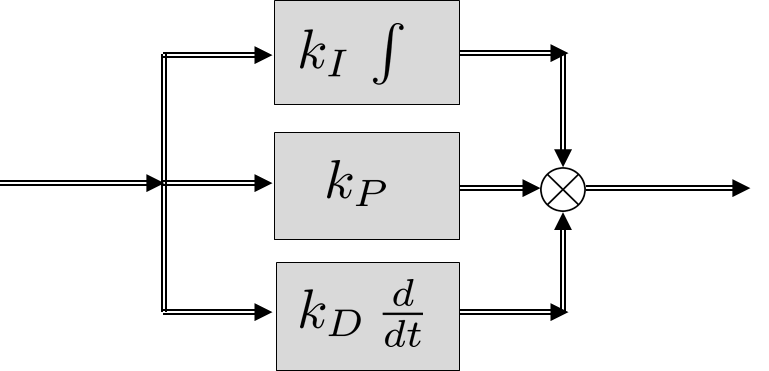}
	\caption{In a PID filter, the input $e(t)$ is transformed into an output $k_P e(t) + k_I \int_0^t e(t')dt' +
	k_D \dot{e} (t)$. The proportional term is the present error, the integral term is the accumulated past error, while the differential term makes account of what may happen in the future. }
	\label{fig:PID_filter}
\end{figure}

\section{Quantum Filtering}
\label{sec:filter}

\subsection{Classical Filtering \& its Quantum Version}
The goal of filtering theory is to make an optimal estimate of the state of a system under conditions where the system may have a noisy dynamics, and where we may have only partial information which itself is subject to observation noise.

\subsubsection{The Kalman Filter (Discrete Time)}
Consider a system whose state belongs to some configuration space $\Gamma$ - say equal to $\mathbb{R}^m$ for some $m$. The state at time $k$ is $x_k \in \Gamma$ and this is supposed to taken to evolve according to a linear but noisy dynamics of the form
\begin{eqnarray}
  x_{k} =A\,x_{k-1}+B\,u_{k}+w_{k}; \qquad (\text{one-step state update rule with noise}).
\end{eqnarray}
Here, the $u_k$ are controls and the $w_k$ are errors. Furthermore, our knowledge is based on measurements that reveal only partial information about the state and the  measurement process may also be corrupted by noise: specifically, we measure $ y_k$ at time $k$ where
\begin{eqnarray}
 y_k =H\,x_{k}+v_{k}; \qquad (\text{partial observations with noise}).
\end{eqnarray}
The random variables $\left( w_{k},v_{k}\right) _{k}$ are the dynamical noise and measurement noise, and are modelled as independent with variables with $w_{k}\sim N\left( 0,Q\right) $ and $v_{k}\sim N\left( 0,1\right) $. If the matrix $H$ may be invertible, then we can fully determine the state from the measurements, at least up to measurement disturbance error, but in general this need not be the case.

\begin{figure}[htbp]
	\centering
		\includegraphics[width=1.0\textwidth]{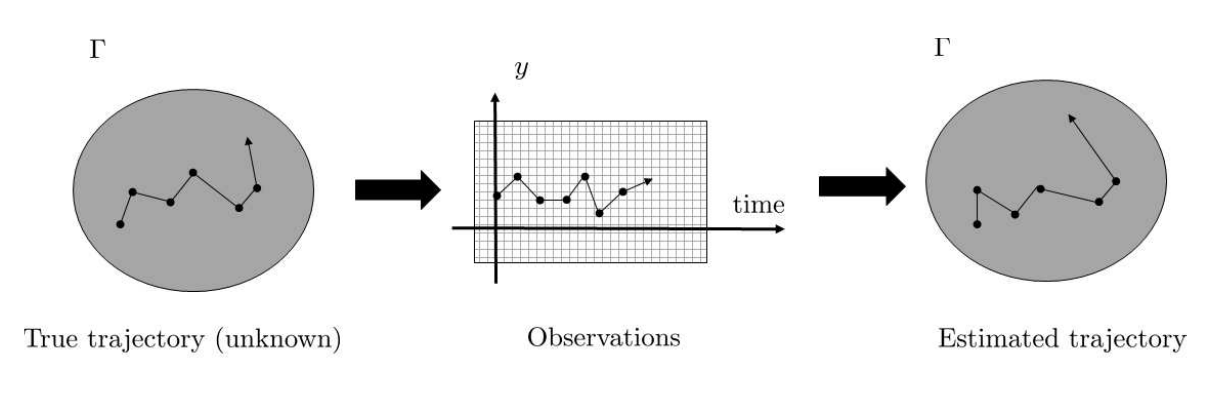}
		\caption{The system has state $x_k$ (unknown) in the configuration space $\Gamma$ at time $k$. Our knowledge of the system
		is based on observations $ y_k$ made at time $k$ containing partial information of the state. Based on our observations
		we estimate what the trajectory might have be.}
	\label{fig:KF_1}
\end{figure}

Our aim is to determine the optimal estimate $ \widehat{f(x_k )}$ for any function $f$ of the state variable $x_k$ 
conditional on the observations $ y_1 , \cdots ,  y_k$ up to time $k$, that is
\begin{eqnarray}
\widehat{f(x_k )} = \mathbb{E} [ f( x_k ) \mid  \mathscr{Y}_{k]} ]
\end{eqnarray}
where $\mathscr{Y}_{k]}$ is the sigma -algebra generated by $ y_1 , \cdots ,  y_k$. Let us introduce the optimal estimate of $x_{k}$  given $y_{1},\cdots , y_{k} $,
\begin{eqnarray}
 \hat{x}_{k} &=&   \mathbb{E} [  x_k  \mid  \mathscr{Y}_{k]} ], \qquad (\text{optimal estimate})
\end{eqnarray}
and also the optimal prediction of $x_{k}$  one time step earlier $\tilde{x}_{k}  =\mathbb{E} [  x_k  \mid  \mathscr{Y}_{k-1]} ]$. The estimation error is $ P_{k} =Var ( x_{k} - \hat{x}_k  ) $, and the prediction error is $\tilde{P}_{k} =Var ( x_{k}- \tilde{x}_k ) $.

The Kalman filter is a \textit{recursive} estimation scheme valid for linear Gaussian models.
Immediately after time $k-1$, we will have recorded the observations $ y_1 , \cdots , y_{k-1}$. Suppose that we have the best estimate $\hat{x}_{k-1}$ at that point, then we would predict $\tilde{x}_{k}=A\,\hat{x}_{k-1}+B\,u_{k}$. From this, we find the predicted state to be $\tilde{x}_{k}=A\,\hat{x}_{k-1}+B\,u_{k}$. Therefore, we predict what the next observation should be $ \tilde{y}_{k}  = H \, \tilde{x}_{k}$. At this stage we wait for the next measured value, $y_k$, to be made at time $ k$. What will be of interest is the \textit{residual error} which is the difference between the \textit{observed value}  of $ y_k$ and the \textit{expected value}, 
$\tilde{y}_k = H\,\tilde{x}_k$, we would have predicted \textit{a priori}. This is quantified as the \textit{innovations process}
\begin{eqnarray}
  I_k =y_k - \tilde{y}_k = y_k- H\, \tilde{x}_{k}, \qquad (\text{innovations}).
\end{eqnarray}
If this error is zero, then our prediction is spot on, and we just take our predicted $ \tilde{x}_{k}$
to be our estimate for state at time $k$. Otherwise, we apply an adjustment based on the discrepancy, and the Kalman filter sets $\hat{x}_{k}=\tilde{x}_{k}+H\tilde{P}_{k}\, I_{k}$ with the associated estimation error $P_{k}=\left( 1-H^{2}\tilde{P}_{k}\right) \tilde{P}_{k}$.

The Kalman filter is therefore a recursive procedure requiring minimal storage of information - effectively only the previous estimate $\hat{x}_{k-1}$ - 
and updated on the next available measurement $y_k$.  It is optimal in the sense that the conditional expectation (uniquely!) provides the least squares estimator for $f(x_k)$ for any
bounded measurable $f$.

\subsubsection{Nonlinear Continuous-time Kalman Filter}
There is a straightforward generalization of the Kalman filter to continuous time, known as the Kalman-Bucy filter.
We will derive it shortly. It is again recursive, but attempts to move away from Gaussian noise or linear dynamics 
result in filters that are difficult to implement. A popular saying in the field was (and arguably still is)
\textit{If we can put a man on the moon, why can't we find a nonlinear filter?}

Let us begin by considering a continuous-time nonlinear-dynamical extension of the estimation problem where the state undergoes a diffusion.
For simplicity we take $\Gamma = \mathbb{R}$, and the continuous time equations are taken as
\begin{eqnarray*}
dX_{t} &=&v\left( X_{t}\right)dt+\sigma \left( X_{t}\right) dW_{t}, \qquad (\text{dynamics});\\
dY_{t} &=&h\left( X_{t}\right) dt+dV_{t} , \qquad (\text{observations});
\end{eqnarray*}
with independent Wiener processes $W_t$ (dynamical noise) and $V_t$ (observational noise).
If we set $f_{t}\equiv f\left( X_{t}\right) $, then from the It\={o} calculus $df_{t}=\mathscr{L}f_{t}\,dt+\sigma _{t}f_{t}^{\prime }\,dW_{t} $,
where the generator of the diffusion process is $ \mathscr{L}f=vf^{\prime }+\frac{1}{2}\sigma ^{2}f^{\prime \prime }$.

The optimal estimate of $f_{t}$, given the observations up to time $t$, will be the conditional; expectation 
$\hat{f}_{t}=\mathbb{E}\left[ f\left( X_{t}\right) |\mathscr{Y}_{t}\right] $, and this will satisfy the \emph{Kushner-Stratonovich equation}
\begin{eqnarray}
d\widehat{f}_{t}=\widehat{\mathscr{L}f}_{t}\,dt+\left\{ \widehat{fh}%
_{t}-\widehat{f}_{t}\widehat{h}_{t}\right\} dI_{t}
\end{eqnarray}
where the \textit{innovations process} is 
\begin{eqnarray}
dI_{t}=dY_{t}-\widehat{h}_{t}\,dt.
\end{eqnarray}
The filter enjoys the least squares property that
\begin{eqnarray}
\mathbb{E} \big[ | \hat{f}_t - f_t |^2 \big] \leq \mathbb{E} \big[ | M - f_t |^2 \big], \qquad (\forall M
\;
\mathscr{Y}_{t]} - \mathrm{measurable}).
\label{eq:LS}
\end{eqnarray}

Statistically, the innovations process will be a Wiener process. This fact underscores the property that the filter is so efficient at extracting information from the measurement process, that afterwards we are just left with white noise! We now give a quick derivation of the Kushner-Stratonovich equation.

\begin{demonstration}[The Kushner-Stratonovich equation]
We use the characteristic function method. Let $\alpha (\cdot )$ be an arbitrary function of time, and define a stochastic
process $C$ by the stochastic differential equation
\begin{eqnarray}
dC_{t}=\alpha \left( t\right) C_{t}dY_{t}
\end{eqnarray}
with $C_{0}=1$. We now make an ansatz that $\hat{f}_{t}$ satisfies a
stochastic differential equation of the form
\begin{eqnarray}
d\hat{f}_{t}=\beta _{t}dt+\gamma _{t}dY_{t}
\end{eqnarray}
where $\beta $ and $\gamma $ are stochastic processes adapted to the
measurement process $Y$.

The least squares property (\ref{eq:LS}) now implies the identity
\begin{eqnarray*}
\mathbb{E}\left[ \left( f_{t}-\hat{f}_{t}\right) C_{t}\right] =0
\end{eqnarray*}
and differentiating wrt. time gives $I+II+III=0$ where
\begin{eqnarray*}
I &=& \mathbb{E}\left[ \left( df_{t}-d\hat{f}_{t}\right) C_{t}\right]  =\mathbb{E}\left[ \left( \mathscr{L}f_{t}-\beta _{t}-\gamma
_{t}h_{t}\right) C_{t}\right] dt, \nonumber\\
II &=&\mathbb{E}\left[ \left( f_{t}-\hat{f}_{t}\right) dC_{t}\right] =\mathbb{E}\left[ \left( f_{t}-\hat{f}_{t}\right) \alpha \left( t\right)
h_{t}C_{t}\right] dt, \nonumber \\
III &=&\mathbb{E}\left[ \left( df_{t}-d\hat{f}_{t}\right) dC_{t}\right]  =\mathbb{E}\left[ -\gamma _{t}\alpha \left( t\right) C_{t}\right] dt.
\end{eqnarray*}

As $\alpha $ was arbitrary, we deduce that $I\equiv 0$ and $II+III\equiv 0$.
The first of these implies that $ \mathbb{E}\left[ \left( \mathscr{L}f_{t}-\beta _{t}-\gamma _{t}h_{t}\right) C_{t}\right] =0$, and by taking conditional expectation under the expectation, and noting that 
$\beta ,\gamma $ and $C$ are adapted,  we get $\beta _{t}=\widehat{\mathscr{L}f_{t}}-\gamma _{t}\hat{h}_{t}$.

Similarly we have $\mathbb{E}\left[ \left( f_{t}-\hat{f}_{t}\right) \alpha \left( t\right)
h_{t}C_{t}\right] -\mathbb{E}\left[ \gamma _{t}\alpha \left( t\right) C_{t}\right] =0$,
and by the same procedure we deduce that $\gamma _{t}=\widehat{f_{t}h_{t}}-\hat{f}_{t}\hat{h}_{t}$.
\end{demonstration}

\subsubsection{The Conditional Density Form}

The filtering problem admits a \textit{conditional density} $\rho_t $ if we have
\[
\widehat{f}_t ( \omega) =\int f(x) \rho_t (x,\omega ) \, dx,
\]
in which case the Kushner-Stratonovich equation becomes 
\begin{eqnarray}
d\rho_t (x ,\omega)  =\mathscr{L}^{\prime } \rho_t (x ,\omega) \,  dt + \left\{ h (x)-\int h (x' ) \rho _t (x',\omega )dx' \right\} \, \rho_t (x ,\omega)\,  dI_{t}.
\label{eq:KS}
\end{eqnarray}

We remark that $\rho_t$ is a probability density function valued stochastic process, and that (\ref{eq:KS}) is in \textit{nonlinear} in $\rho $! 
 If we average over all outputs we get an average density $\bar{ \rho}_t$ which satisfies the Fokker-Planck (Kolmogorov forward) equation
\begin{eqnarray}
\frac{d}{dt} \bar{ \rho }_{t}=\mathscr{L}^\star  \bar{\rho }_t,
\end{eqnarray}
where the dual of the diffusion generator is $\mathscr{L}^\star g = ( vg )^\prime + \frac{1}{2} ( \sigma^2 g )^{\prime \prime} $.

An equivalent formulation is the \textit{Zakai equation} for an \emph{unnormalized} stochastic density function $\xi _{t}$,
and is the \textit{linear} equation
\begin{eqnarray}
d\xi _{t}=\mathscr{L}^{\prime }\xi _{t}\,dt+\left( \sigma \xi
_{t}\right) ^{\prime }\,dY_{t}.
\end{eqnarray}
We then have $\rho _{t}(x, \omega ) =\xi _{t}(x, \omega ) /\int \xi _{t} (x', \omega ) dx'$.

\subsubsection{The Kalman-Bucy Filter}

As a special case, take 
\begin{eqnarray*}
dX_{t} &=&AX_{t}\,dt +Bu(t) \, dt +dW_{t} \\
dY_{t} &=&HX_{t}\,dt+dV_{t}
\end{eqnarray*}
so that $dI_{t}=dY_{t}-H\hat{X}_{t}dt$.  We have 
\begin{eqnarray*}
d\widehat{X}_{t} &=&A\widehat{X}_{t}\,dt + Bu(t) \, dt +H\left( \widehat{X_{t}^{2}}-%
\widehat{X_{t}}\widehat{X_{t}}\right) \,dI_{t}, \\ 
d\widehat{X_{t}^{2}} &=& (2A\widehat{X_{t}^{2}} +2Bu(t) \hat{x}_t +Q) dt+H\left( \widehat{X_{t}^{3}}-%
\widehat{X_{t}^{2}}\widehat{X_{t}}\right) dI_{t} .
\end{eqnarray*}

From Gaussianity,
we have $\widehat{\big( X_{t}-\widehat{X_{t}}\big) ^{3}}=\widehat{X_{t}^{3}}-3\widehat{X}_{t}\widehat{X_{t}^{2}} +2\widehat{X_{t}}^{2}=0$,
and this allows us to reduce the order of the problem so that we only have to contend with moments and variances of the conditional distribution.

Let $P_{t}=\widehat{X_{t}^{2}}-\widehat{X_{t}}^{2}$, then we are lead to the Kalman-Bucy filter
\begin{eqnarray*}
d\hat{X}_{t}&=&A\widehat{X}_{t}\,dt +Bu(t) \, dt+HP_{t}\,dI_t \\
&\equiv & (A-H^{2}P_{t})\widehat{X}_{t}\,dt +Bu(t) \, dt +HP_{t}\,dY_{t}, 
\\ 
\frac{d}{dt}P_{t}&=&2AP_{t}+Q-H^{2}P_{t}^{2}.
\end{eqnarray*}

\subsection{Physical Motivation Of Quantum Filters}
The theory of quantum filtering was developed by V.P. Belavkin \cite{Belavkin1}, and represents the continuation of the work of Kalman, Stratonovich, Kushner, Zakai, etc.
Partial information about the state of high $Q$ cavity modes is obtained by measuring Rydberg atoms
that are passed one-by-one through the cavity, see Figure \ref{fig:LKB_Photon_Box}.

\begin{figure}[h]
	\centering
		\includegraphics[width=0.50\textwidth]{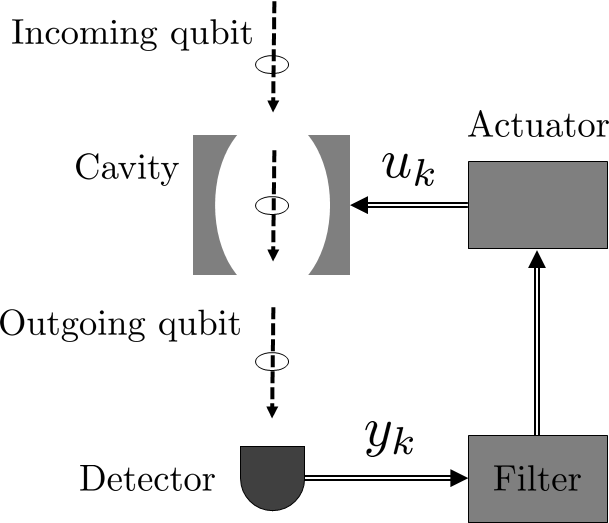}
	\caption{A schematic of the LKB photon box experiment: qubits (Rydberg atoms) are passed 
	through a cavity one by one. At any one time there will be at most one qubit inside the cavity,
	and we measure the outgoing qubits one-by-one. The measurement bit $y_k$ is sent into a filter
	which estimates the state of the cavity mode, and then an instruction is sent to the actuator
	so as to control the mode.}
	\label{fig:LKB_Photon_Box}
\end{figure}

The measurement results can be used to apply a feedback action on the cavity mode.
We consider a quantum mechanical system which is probed by a two-level atom
(qubit). The qubit is in input state $|\downarrow \rangle $ initially. The
unitary interaction between the cavity mode and the probe qubit leads to a change of state
in the Schr\"{o}dinger picture:
\begin{eqnarray}
|\psi \rangle \otimes |\downarrow \rangle \longrightarrow U\,|\psi \rangle
\otimes |\downarrow \rangle .
\end{eqnarray}
We take the interaction time $\tau $ to be very small and assume that the
unitary has the form
\begin{eqnarray*}
U &=&\exp \left\{ \sqrt{\tau }L\otimes \sigma ^{\ast }-\sqrt{\tau }L^{\ast
}\otimes \sigma -i\tau H\otimes I_{2}\right\}  \\
&\simeq &1+\sqrt{\tau }L\otimes \sigma ^{\ast }-\sqrt{\tau }L^{\ast }\otimes
\sigma   -\tau (\frac{1}{2}L^{\ast }L+iH)\otimes I_{2}+\cdots .
\end{eqnarray*}

We now measure the spin $\sigma_{x}$ of the qubit and record the eigenvalues $\eta =\pm 1$ corresponding to eigenvectors
\begin{eqnarray*}
|+\rangle =\frac{1}{\sqrt{2}}|\downarrow \rangle +\frac{1}{\sqrt{2}}%
|\uparrow \rangle ,\qquad |-\rangle =\frac{1}{\sqrt{2}}|\downarrow \rangle -%
\frac{1}{\sqrt{2}}|\uparrow \rangle .
\end{eqnarray*}
The probabilities for detecting $\eta =\pm 1$ are
\begin{eqnarray*}
p_{\pm }=\frac{1}{2} \pm \frac{1}{2} \sqrt{\tau }\langle \psi |L+L^{\ast }|\psi
\rangle +\cdots .
\end{eqnarray*}
After measurement, the system state becomes (up to normalisation!)
\begin{eqnarray*}
|\psi _{\eta }\rangle \propto |\psi \rangle +\sqrt{\tau }L|\psi \rangle
\,\eta -\tau (\frac{1}{2}L^{\ast }L+iH)|\psi \rangle +\cdots .
\end{eqnarray*}
For more details on how this is interpreted as a discrete-time quantum Kalman filter, see \cite{Rouchon}.

We may now proceed to take a continuous time limit where we have (due to a central limit effect)
\begin{eqnarray*}
\tau \hookrightarrow dt\quad \sqrt{\tau }\eta \hookrightarrow dY(t)
\end{eqnarray*}
where $Y(t)$, the continuous time measurement readout, will be a diffusion process. The limit equation is 
\begin{eqnarray}
d|\chi _{t}\rangle =L|\chi _{t}\rangle \,dY_t -(\frac{1}{2}L^{\ast }L+iH)|\chi
_{t}\rangle \,dt ,
\label{eq:Belavkin_Zakai}
\end{eqnarray}
and we refer to this as the \textit{Belavkin-Zakai equation} as it plays the same role as the Zakai equation
in Belavkin's theory of quantum filtering. In general $| \chi_t \rangle$ is not normalized, but it is easy to
obtain the equation for the normalized state $|\psi _{t}\rangle =|\chi _{t}\rangle /\left\|
\chi _{t}\right\| $, we find the \textit{Stochastic Schr\"{o}dinger equation}
\begin{eqnarray}
d|\psi _{t}\rangle =
iH|\psi _{t}\rangle \,dt 
- \frac{1}{2} (L -\lambda _{t})^\ast (L - \lambda _{t}) |\psi _{t}\rangle
\,dt  +
(L-\lambda _{t})|\psi _{t}\rangle \,dI (t) .
\end{eqnarray}
where
\begin{eqnarray}
\lambda _{t} \triangleq \langle \psi _{t}|L+L^{\ast }|\psi _{t}\rangle , \quad dI (t) \triangleq dY_{t}-\lambda _{t}dt .
\end{eqnarray}
Mathematically $I(t)$ has the statistics of a Wiener process, and its increment $dI(t)$ is the difference between what we observe, $dY(t)$,
and what we would expect to get $\langle \psi _{t}|L+L^{\ast }|\psi _{t}\rangle \, dt$.

It is convenient to frame this in the Heisenberg picture. Let $| \psi_0 \rangle$ be the initial state, then setting 
$\langle \psi_0 |\widehat{X}_{t}| \psi_0 \rangle \triangleq \langle \psi _{t}|X|\psi _{t}\rangle $,
we have
\begin{eqnarray*}
d\widehat{X}_t &=& \widehat{\mathscr{L}X}_t \, dt
+\left\{ \widehat{XL}_t-\widehat{X}_t \widehat{L}_t \right\} dI (t) 
+\left\{ \widehat{L^\ast X}_{t}
-\widehat{L^\ast }_{t}\widehat{X}_{t}\right\} dI (t) \\
&=& \widehat{\mathscr{L}X}_{t} \, dt 
+\left\{ \widehat{XL+L^\ast X}_{t}- \widehat{X}_{t}(\widehat{L+L^\ast}_t ) \right\} dI (t) ,
\end{eqnarray*}
where the dynamical part involves the GKS-Lindblad generator
\[ 
\mathscr{L} X\triangleq \frac{1}{2}\left[ L^{\ast },X\right] L+\frac{1}{2}L^{\ast }\left[ X,L\right] -i\left[ X,H\right] 
\]
and the innovations may be written as
\[
dI_t = dY_t - ( \widehat{L+L^\ast}_t ) \, dt .
\]

We may set $\widehat{X_t} = \mathrm{tr} \{ \varrho_t X \}$ and obtain the \textit{Stochastic Master equation} 
\begin{eqnarray*}
d \varrho_t =  \mathscr{L}^\prime \varrho_t  \, dt
+  \left\{ L \varrho_t + \varrho_t L^\ast - \mathrm{tr} \{ \varrho_t (L+L^\ast )\varrho_t  \right\} \, dI_t .
\end{eqnarray*}
where we have the dual generator
$$\mathscr{L}^\prime \varrho  =   L \varrho  L^\ast  -\frac{1}{2} \varrho  L^\ast L - \frac{1}{2} L^\ast L \varrho  +i[ \varrho
 , H ] .$$
This is the quantum version of the conditional probability density form of the Kushner-Stratonovich equation.
Averaging over the output yields the \textit{Master equation}
\[
\frac{d}{dt} \bar{\varrho}_t =   L \bar{\varrho}_t L^\ast  -\frac{1}{2} \bar{\varrho}_t L^\ast L - \frac{1}{2} L^\ast L \bar{\varrho}_t +i[ \bar{\varrho}_t , H ] .
\]
This is the analogue of the Fokker-Planck equation.

\subsection{Quantum Probability}
We wish to present quantum filtering as a natural extension of the classical theory. To this end, it is extremely helpful to extend the standard theory
of probability (the Kolmogorov framework) to that of quantum probability. here we follow closely the lucid exposition given by Maassen \cite{Maassen88}.

\begin{definition}
A \textbf{Quantum Probability space}, or \textbf{QP space}, is a pair $( \mathfrak{A}, \mathbb{E} )$ consisting of
a von Neumann algebra $\mathfrak{A}$ with normal state $\mathbb{E}$.
\end{definition}

The use of the term normal here refers to continuity in the normal topology and has nothing to do with gaussianity.
The algebra of bounded functions $L^{\infty }(\Omega ,\mathscr{F},\mathbf{P})$ on a classical probability space is a commutative example,
with $\mathbb{E} [f] = \int_\Omega f(\omega ) \, \mathbf{P} [d\omega ]$. At its heart, quantum probability amounts to dropping the commutativity
assumption and this leads to the situation in quantum theory where we now have a closed algebra of operators on a Hilbert space $\mathfrak{h}$, with
the normal states  taking the form $\mathbb{E} [X] \equiv \mathrm{tr} \{ \varrho X \}$, for some density matrix $\varrho$.

\begin{definition} 
An \textbf{observable} in the QP space $( \mathfrak{A} , \mathbb{E})$ is a self-adjoint operator $X$ with $\mathfrak{A}$-valued 
spectral measure $\Pi$. That is,
\[
f(X) \equiv \int_{\mathbb{R}} f(x) \, \Pi [dx] ,
\]  
where $\Pi [dx]$ is a projection-valued measure with $\Pi [A] \in \mathfrak{A}$ the event that $X$ takes a value in (Borel) set $A$.
\end{definition}

\begin{itemize}

\item The probability distribution of $X$ in the state $\mathbb{E} [\cdot ] \equiv \mathrm{tr} \{ \varrho \cdot \}$ is
\[
\mathbb{K} [dx] =\mathrm{tr} \{ \varrho \, \Pi [ dx ]\}.
\]

\item Quantum probabilists would say that the observable is a CP embedding $j$ of the commutative QP space $\mathfrak{A} = 
L^\infty ( \mathbb{R} , \mathbb{K} )$  into $\mathfrak{A}$, that is $:f \mapsto j(f) \triangleq f(X)$.

\item More generally we define a \textbf{quantum random variable} to be a CP embedding of a \textit{non-commutative} QP space $\mathfrak{B}$
into $\mathfrak{A}$.

\end{itemize}

\bigskip

A quantum probability space for a qubit is therefore the algebra of $2 \times 2$ matrices with a suitable density matrix as state.It is an important, though inconvenient, fact that the position and momentum operators in quantum mechanics are unbounded operators. Technically they do not themselves belong to the von Neumann algebra they generate. The standard approach to generating a von Neumann algebra from a set of observables is to take the closure of the algebra generated by the corresponding projections: this is the analogue of generating a $\sigma$-algebra of subsets given a set of random variables in classical probability (measure) theory.

In general, for von Neumann sub-algebras $\mathfrak{B}_k$ of a von Neumann algebra $\mathfrak{A}$, we will write $ \mathfrak{B}_1 \vee \mathfrak{B}_2$ for the smallest
von Neumann sub-algebra of $\mathfrak{A}$ containing both algebras. For $X_1, \cdots , X_n$ a collection of observables, we denote the
von Neumann algebra they generate by $\mathrm{vN} \{ X_1, \cdots , X_n\}$.

\begin{definition}
A \textbf{conditional expectation} is a CP map from a QP space $( \mathfrak{A}, \mathbb{E})$ onto
a sub-QP space $(\mathfrak{Y} , \mathbb{E} \mid_{\mathfrak{Y}})$ such that $\mathbb{E}[\widehat{X}] = \mathbb{E} [X]$ which is a \textit{projection}, that is 
\begin{eqnarray}
\widehat{\hat{X}} =\hat{X} .
\end{eqnarray}
\end{definition}

There is some good news and some bad news. The good news first.

\begin{proposition}
Suppose that we have a quantum conditional expectation, say  $\hat{X}=\mathbb{E} [X \mid \mathfrak{Y} ] $, generated by some quantum random variable, 
then it will satisfy
\begin{enumerate}
\item $\mathbb{E}[Y_{1}\,\hat{X}\,Y_{2}]=\mathbb{E}\left[ Y_{1}\,X\,Y_{2}\right] $, for all $Y_{1}$ and $Y_{2}$ be in $\mathfrak{Y}$;

\item Least squares property: 

$\min_{Y\in \mathfrak{Y}} \mathbb{E}[(X-Y)^{\ast }(X-Y)]$ is given by
\begin{eqnarray*}
 \mathbb{E}[(X-\hat{X})^{\ast }(X- \hat{X} )]
=\mathbb{E}[\widehat{X^{\ast }X}]-\mathbb{E}[\hat{X}^{\ast }\hat{X}].
\end{eqnarray*}
\end{enumerate}
\end{proposition}

Now for the bad news: quantum conditional expectations do not always exist. See \cite{Maassen88} for more discussion.

%%%%%%%%%%%%%%%%%%%%%%

\subsection{Non-Demolition Measurement and Estimation}
Let $(\mathfrak{A} , \mathbb{E} )$ be the QP model of interest,
and suppose $\mathfrak{Y}$ be the von Neumann algebra generated by the observables we measure in an experiment.
We recall the definition of the commutant of an algebra.

\begin{definition}
The \textbf{commutant} of a von Neumann sub-algebra $\mathfrak{Y}$ of a von Neumann algebra $\mathfrak{A}$ is
\begin{eqnarray}
\mathfrak{Y}^\prime \triangleq \{ A \in \mathfrak{A} : [A,Y]=0, \, \forall \, Y \in \mathfrak{Y} \}.
\end{eqnarray}
\end{definition}

The \textbf{Non-Demolition Principle} states that all the measured observables in an experiment should be \emph{compatible}.
That is, $\mathfrak{Y}$ should be \textit{commutative}, and that an observable $X$ is jointly measurable (and estimable) with the measured observables if and only if it commutes with them
($X \in \mathfrak{Y}^\prime $).

A key point is that we can always define conditional expectation of operators in $\mathfrak{Y}^\prime$ onto $\mathfrak{Y}$, see Figure \ref{fig:Venn_Neumann}. On the left we have $\mathfrak{Y}$ which is a commutative
subalgebra of $\mathfrak{A}$. The commutant $\mathfrak{Y}^\prime$ contains $\mathfrak{Y}$, but is generally noncommutative.

\begin{figure}[h]
	\centering
		\includegraphics[width=1.0\textwidth]{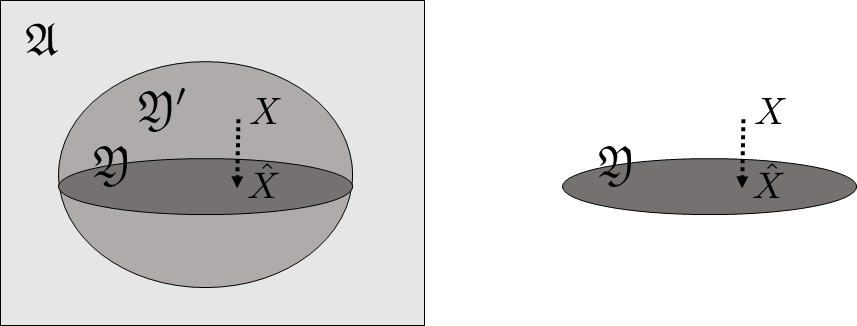}
		\caption{A \lq\lq Venn Neumann\rq\rq diagram!}
	\label{fig:Venn_Neumann}
\end{figure}

The construction is actually straightforward. Suppose that $X \in \mathfrak{Y}^\prime$ then we note that $ \mathfrak{Y} \vee \mathrm{vN} \{X \}$ is a commutative von Neumann algebra:
this is shown on the right side of Figure \ref{fig:Venn_Neumann} where we have stripped away the non-commuting parts of the rest of the universe. However, conditional expectations always (uniquely!)
exist in the classical world, and because $ \mathfrak{Y} \vee \mathrm{vN} \{X \}$ is commutative and so the corresponding QP model is isomorphic to a classical probability model.
The argument in fact leads to a unique value for $\hat X$.

It is however worth mentioning that, despite the simplicity of the above argument, the construction is non-trivial. For instance, we could take a pair of observables $X_1$ and $X_2$ from the
commutant $\mathfrak{Y}^\prime$ and deduce the existence and values of $\hat{X}_1$ and $\hat{X}_2$ - however the commutant is not required to be a commutative algebra, and so $\mathrm{vN} \{ X_1 , X_2 \}$ need not be commutative either.

%%%%%%%%%%%%%%%%%%%%%%%%%%%%%%%%%%%%%%%%%%%%%%%%%%%%%%%%%%%%%%%%%%%%%%%%%%%%%%%%%

\subsection{Quantum Input-Output Models}
This is sometimes referred to as the $SLH$ formalism and is a synthesis of the quantum stochastic calculus developed by Hudson-Parthasarathy in 1984, \cite{HP84,Par92} and the quantum input-output
theory developed by Gardiner and Collett in 1985, \cite{Gardiner_Collett,Gardiner_Zoller}. We work in the setting of a Hilbert space of the generic form 
\begin{eqnarray}
\mathfrak{H}=\mathfrak{h}_{0}\otimes \Gamma \left( \mathfrak{K}\otimes
L^{2}[0,\infty )\right)
\end{eqnarray}
where $\mathfrak{h}_{0}$ is a fixed Hilbert space, called the \emph{initial
space}, and $\mathfrak{K}$ is a fixed Hilbert space called the \emph{internal space} or \textit{multiplicity space}. 
In the case of $n$ bosonic input processes we fix a label set $\mathsf{k}=\left\{ 1,\cdots ,n\right\}$ and take the multiplicity space to be $\mathfrak{K}=\mathbb{C}^{\mathsf{k}}
\equiv \mathbb{C}^n$. For each $t >0$ we have tensor product decomposition
$\mathfrak{H} \cong \mathfrak{H}_{\left[ 0,t\right] }\otimes \mathfrak{H}_{\left( t,\infty \right) }$,
where $\mathfrak{H}_{\left[ 0,t\right] }=\mathfrak{h}_{0}\otimes \Gamma
\left( \mathfrak{K}\otimes L^{2}[0,t)\right) $ and $\mathfrak{H}_{\left(
t,\infty \right) }=\Gamma \left( \mathfrak{K}\otimes L^{2}(t,\infty )\right) 
$. We shall write $\mathfrak{A}_{t]}$ for the von Neumann algebra of bounded operators on $%
\mathfrak{H}$ that act trivially on the future component $\mathfrak{H}%
_{\left( t,\infty \right) }$, that is $\mathfrak{A}_{t]} \equiv \mathscr{B} ( \mathfrak{H}_{\left[ 0,t\right] }) $.

The general form of the \emph{constant} operator-coefficient quantum stochastic
differential equation for an adapted unitary process $U$ is 
\begin{eqnarray}
dU\left( t\right) &=& \bigg\{ -\left( \frac{1}{2}L_{\mathsf{k}}^{\ast }L_{%
\mathsf{k}}+iH\right) dt+\sum_{j\in \mathsf{k}}L_{j}dB_{j}\left( t\right)
^{\ast }  \nonumber \\
&&-\sum_{j,k\in \mathsf{k}} L_j^\ast S_{jk} dB_{k}\left( t\right) +\sum_{j,k\in 
\mathsf{k}}(S_{jk}-\delta _{jk})d\Lambda _{jk}\left( t\right) \bigg\} U\left( t\right)
\label{eq:Ito_QSDE}
\end{eqnarray}
where the $S_{jk},L_{j}$ and $H$ are operators on the initial Hilbert space
which we collect together as
\begin{eqnarray}
S_{\mathsf{kk}}=\left[ 
\begin{array}{ccc}
S_{11} & \cdots & S_{1n} \\ 
\vdots & \ddots & \vdots \\ 
S_{n1} & \cdots & S_{nn}
\end{array}
\right] , \quad L_{\mathsf{k}}=\left[ 
\begin{array}{c}
L_{1} \\ 
\vdots \\ 
L_{n}
\end{array}
\right] .
\end{eqnarray}

The necessary and sufficient conditions for the process $U$ to be unitary are that 
$S_{\mathsf{kk}}=\left[ S_{jk}\right] _{j,k\in \mathsf{k}}$ is unitary (i.e., $\sum_k S_{kj}^\ast S_{kl} = \sum_k S_{jk}S_{lk}^\ast = \delta_{jl} \, I$)
and $H$ self-adjoint. (We use the convention that $L_{\mathsf{k}}=\left[
L_{k}\right] _{k\in \mathsf{k}}$ and that \ $L_{\mathsf{k}}^{\ast }L_{%
\mathsf{k}}=\sum_{k\in \mathsf{k}}L_{k}^{\ast }L_{k}$.) 
The triple $G \sim \left( S,L,H\right) $ are termed the \textit{Hudson-Parthasarathy parameters}, or informally as just the $SLH$-\textit{coefficients}.
For more details of the $SLH$ formalism see  \cite{GouJam09a,GouJam09b,CKS}.

We will sketch the situation where we have an open quantum system $G \sim \left( S,L,H\right) $ driven by boson fields as a block component as in Figure \ref{fig:eye}.
To complete our theory, we need to say how the outputs can be incorporated into the framework. The output fields are, in fact, defined by
\begin{eqnarray}
B^{\text{out}}_k (t) \triangleq U ( t)^\ast \left[ I\otimes B_k (t) \right] U\left( t\right) ,
\label{eq:output}
\end{eqnarray}
and a simple application of the quantum It\={o} calculus yields
\begin{eqnarray}
dB_{j}^{\text{out}}\left( t\right) =\sum_{k}j_{t}\left( S_{jk}\right)
\,dB_{k}\left( t\right) +j_{t}\left( L_j\right) \,dt .
\end{eqnarray}
The reasoning behind (\ref{eq:output}) is given by Gardiner and Collett \cite{Gardiner_Collett}, see also 
\cite{Gardiner_Zoller}.

\begin{figure}
	\centering
		\includegraphics[width=0.60\textwidth]{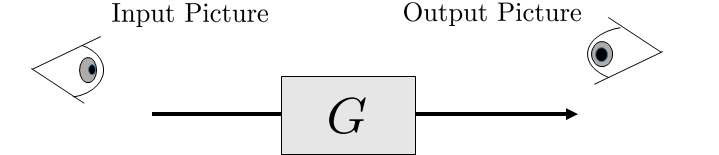}
		\caption{SLH block component: we have quantum inputs coming in from the left and quantum outputs going out on the right, and these are described by arrows
		(they may generally carry a multiplicity) while the system is denoted by a block and the $SLH$-coefficients are $G \sim \left( S,L,H\right) $.}
	\label{fig:eye}
\end{figure}

Let $X$ be an initial space operator, then we consider its time evolution which we write as $j_{t}\left(
X\right) \triangleq U\left( t\right) ^{\ast }\left[ X\otimes I\right]
U\left( t\right) $.
From the quantum It\={o} calculus, we obtain the Heisenberg-Langevin equation
\begin{eqnarray}
dj_{t}\left( X\right) &=& j_{t}\left( \mathscr{L}X\right)
dt+\sum_{i}j_{t}\left( \mathscr{L}_{i0}X\right) dB_{i}^{\ast }\left( t\right) \nonumber \\
&&
+\sum_{i}j_{t}\left( \mathscr{L}_{0i}X\right) dB_{i}\left( t\right)
+\sum_{j,k}j_{t}\left( \mathscr{L}_{jk}X\right) d\Lambda _{jk}\left(
t\right) ;
\label{eq:heisenberg}
\end{eqnarray}
where 
\begin{eqnarray}
\mathscr{L}X &=&\frac{1}{2}\sum_{i}L_{i}^{\ast }\left[ X,L_{i}\right] +\frac{%
1}{2}\sum_{i}\left[ L_{i}^{\ast },X\right] L_{i}-i\left[ X,H\right] , \, 
\text{(the Lindbladian!)},\nonumber \\
\mathscr{L}_{i0} X &=& \sum_j S_{ji}^\ast [X,L_j ], \nonumber \\
\mathscr{L}_{0i} X &= & \sum_k [L_k^\ast , X ] S_{ki},\nonumber \\
\mathscr{L}_{ik} X &=& \sum_j S_{ji}^\ast XS_{jk} - \delta_{ik} X .
\label{eq:EH}
\end{eqnarray}

We may write this as 
\begin{eqnarray*}
dj_{t}\left( X\right) =\sum_{\alpha =0}^{n}\sum_{\beta =0}^{n}j_{t}\left( %
\mathscr{L}_{\alpha \beta }X\right) \, d\Lambda ^{\alpha \beta }\left(
t\right)
\end{eqnarray*}
where $\mathscr{L}_{00}=\mathscr{L}$ and $\Lambda ^{00}\left( t\right) =t$, $%
\Lambda ^{j0}\left( t\right) =B_{j}\left( t\right) ^{\ast }$ and $\Lambda
^{0k}\left( t\right) =B_{k}\left( t\right) $.

We now make a crucial observation. The unitary process $U$ is responsible for coupling the system to the boson field environment
and does so in a causal way. In particular, the algebra $\mathfrak{A}_{t]}$ consists of precisely the initial system together with the
component of the environment that has had a chance to interact with it up to time $t$. The unitary process is also responsible for 
transferring us from the input picture to the output picture. At the moment, this distinction is fairly straightforward but it will become
more involved when we allow the $SLH$-coefficients to be adapted process rather than just fixed operators on the initial space.

The output field may, of course, then be measured. In a homodyne measurement, we may
measure the quadrature process $Y_k (t) = B_k^{\text{out}} (t) +B^{\text{out} \ast}_k(t)$. Note that in this case,
\begin{eqnarray}
dY_k (t) = \sum_{k}j_{t}\left( S_{jk}\right) \,dB_{k}\left( t\right) +\sum_{k}j_{t}\left( S_{jk}^\ast \right) \,dB_{k}^\ast\left( t\right) 
+j_{t}\left( L_j+ L_j^\ast \right) \,dt .
\label{eq:standard_quadrature_output}
\end{eqnarray}
We note that $Y_k (t)$ is obtained from $Z_k (t) = B_k (t)^\ast + B_k (t)$ by unitary conjugation by $U(t)$. As the $Z_k$ is a self-commuting
process, we have that so too will $Y_k $ - an essential requirement for non-demolition measurement. the same would be true if we wanted to 
count photons at the output: we would set $Z_k (t) = \Lambda_{kk} (t)$, which self-commuting for different $t$, and take $Y_k$ to the 
equivalent process obtained by transferring from the input to the output picture.

\begin{definition}[Measured Process]
Let $Z_t$ be an adapted self-commuting process on the Fock space. The corresponding \textbf{measured process} at output is
\[
Y_t \triangleq U_t^\ast (1 \otimes Z_t ) U_t ,
\]
and its measurement algebra is defined to be
\[
\mathfrak{Y}_t = \mathrm{vN} \{ Y_s : 0 \leq s \leq t \} ,
\]
\end{definition}

In general, we could have several adapted processes $Z_k$ on the Fock space - all commuting with each other - and measure the corresponding output picture processes
$Y_k (t)$. We remark that $\mathfrak{Y}_t$ will then be a \emph{commutative} von Neumann algebra in accordance with the non-demolition principle. In fact, 
we have a \textbf{filtration} of von Neumann algebras $\{ \mathfrak{Y}_t : t \geq 0 \}$ with
$\mathfrak{Y}_t \subset \mathfrak{Y}_s$ whenever $ t<s$.

We also have the following dynamical non-demolition principle: $j_s (X) \in \mathfrak{Y}_t^\prime$ for each $s \geq t$ and for all operators, $X$, on the initial space.
To see this, note that, for $s \geq t$, $Y_t =U_t^\ast (1 \otimes Z_t ) U_t = U_s^\ast (1 \otimes Z_t ) U_s $ while  
$j_s (X) =  U_s^\ast (X \otimes 1 ) U_s$, therefore  $[ j_s (X) , Y_t ] =  U_s^\ast [ X \otimes 1 , 1 \otimes Z_t  ] U_s $.
It follows that present Heisenberg picture value $j_t (X)$ is compatible with the measurement algebra $\mathfrak{Y}_t$, for each $t$.
Therefore, we may estimate (filter) $j_t (X)$ based on the measurements up to current time. The same is also true for future values $j_s (X)$, $s >t$,
which we may then predict using the measurements up to current time.

As we have seen, the conditional expectation of an observable $X$ in the Heisenberg picture onto the measurement output algebra $\mathfrak{Y}_t$ then exists,
and we call this the \emph{filtered estimate} for $X$ based on the measurements. We denote this as
\begin{eqnarray}
\pi_t (X) \triangleq \mathbb{E} [ j_t ( X  ) \mid \mathfrak{Y}_t ] .
\label{eq:pi}
\end{eqnarray}

\subsection{The Belavkin-Kushner-Stratonovich Equation}
We now follow the notation of \cite{BvH}, see also \cite{BouvanHJam07}
As a special example, let us take $Z (t) = e^{i\theta (t)} B(t) + e^{-i \theta (t)} B(t)^\ast$.
The phase $\theta $ may be time dependent, and in principle even an adapted
process though we assume it to be deterministic for the time being.
The measured field is $Y^{\theta }\left( t\right) $ where 
\begin{eqnarray*}
dY^{\theta }\left( t\right) =e^{i\theta \left( t\right) }dB_{\text{out}%
}\left( t\right) +e^{-i\theta \left( t\right) }dB_{\text{out}}\left(
t\right) ^{\ast }.
\end{eqnarray*}
We denote by $\mathfrak{Y}_{t]}^{\theta }$ the von Neumann algebra generated by $%
Y^{\theta }\left( s\right) $ for $0\leq s\leq t$.

The filter is the conditional expectation 
\begin{eqnarray*}
\pi _{t}^{\theta }\left( X\right) =\mathbb{E}\left[ j_{t}\left( X\right) |%
\mathfrak{Y}_{t]}^{\theta }\right]
\end{eqnarray*}
for the state which is a factor state with the field in the Fock vacuum.

\begin{lemma}
The filter satisfies the Belavkin-Kushner-Stratonovich equation. 
\begin{eqnarray}
d\pi _{t}^{\theta }\left( X\right) =\pi _{t}^{\theta }\left( \mathscr{L}%
X\right) dt+\mathscr{G}_{t}^{\theta }\left( X\right) dI (t) ^{\theta }
\label{eq:q_filter}
\end{eqnarray}
where $\mathscr{L}$ is the Lindbladian $\frac{1}{2}L^{\ast }\left[ \cdot ,L%
\right] +\frac{1}{2}\left[ L^{\ast },\cdot \right] L-i\left[ \cdot ,H\right] 
$, and
\begin{eqnarray}
\mathscr{G}_{t}^{\theta }\left( X\right) =  \pi _{t}^{\theta }\left(
e^{i\theta \left( t\right) }XL+e^{-i\theta \left( t\right) }L^{\ast
}X\right) -\pi _{t}^{\theta }\left( X\right) \pi _{t}^{\theta }\left(
e^{i\theta \left( t\right) }L+e^{-i\theta \left( t\right) }L^{\ast }\right)
\label{eq:G(X)}
\end{eqnarray}
and $I^{\theta }$ is the innovations process $dI^{\theta }\left( t\right)
=dY^{\theta }\left( t\right) -\pi _{t}^{\theta }\left( e^{i\theta \left(
t\right) }L+e^{-i\theta \left( t\right) }L^{\ast }\right) dt$.
\end{lemma}

\textbf{Remark:} the filter obtained is precisely the same as for the
problem of filtering the open systems dynamics with $L$ replaced by the time
dependent coupling $e^{i\theta \left( t\right) }L$ and a continuous
measurement of the fixed quadrature $Y^{0}\left( t\right) =B_{\text{out}%
}\left( t\right) +B_{\text{out}}\left( t\right) ^{\ast }$.

\begin{proof}

(We drop the $t$ and $\theta $ identifiers for simplicity.) The proof follows
the same line of argument used to demonstrate the Kushner-Stratonovich equation by
extending the characteristic method and watching out for the noncommutative elements. The ansatz is 
made that the filter satisfies $d\pi \left( X\right) =\mathscr{H}(X)dt+%
\mathscr{G}\left( X\right) dY$ where both $\mathscr{H}(X)$ and $\mathscr{G}%
\left( X\right) $ are in $\mathfrak{Y}_{t]}^{\theta }$. We have that 
\begin{eqnarray*}
\mathbb{E}\left[ \left( \pi \left( X\right) -j\left( X\right) \right) y\right] =0  
\label{eq:ortho}
\end{eqnarray*}
for all $y\in \mathfrak{Y}_{t]}^{\theta }$. In particular, we set $y=y_{t}$
where $dy_{t}=f\left( t\right) y_{t}dY^{\theta }\left( t\right) $ with $f$
an arbitrary integrable function. Taking the It\={o} differential of (\ref
{eq:ortho}) we find $I+II+III=0$ where 
\begin{eqnarray*}
I &=&\mathbb{E}\left[ \left( d\pi \left( X\right) -dj\left( X\right) \right)
y\right] \\
&=&\mathbb{E}\left[ \left( \mathscr{H}(X)dt+\mathscr{G}\left( X\right)
\left( e^{i\theta \left( t\right) }L+e^{-i\theta \left( t\right) }L^{\ast
}\right) -\mathscr{L}X\right) y\right] dt \\
II &=&\mathbb{E}\left[ \left( \pi \left( X\right) -j\left( X\right) \right)
dy\right] \\
&=&f\left( t\right) \mathbb{E}\left[ \left( \pi \left( X\right) -j\left(
X\right) \right) \left( e^{i\theta \left( t\right) }L+e^{-i\theta \left(
t\right) }L^{\ast }\right) y\right] dt
\end{eqnarray*}
and the It\={o} correction is 
\begin{eqnarray*}
III=\mathbb{E}\left[ \left( d\pi \left( X\right) -dj\left( X\right) \right)
dy\right] =f\left( t\right) \mathbb{E}\left[ \left( \mathscr{G}\left(
X\right) -e^{-i\theta \left( t\right) }L^{\ast }\right) y\right] dt.
\end{eqnarray*}
Taking the coefficient of $f\left( t\right) dt$ leads to (\ref{eq:G(X)})
while the remain coefficient leads to 
\begin{eqnarray*}
\mathscr{H}(X)=\mathscr{G}\left( X\right) \pi \left( e^{i\theta \left(
t\right) }L+e^{-i\theta \left( t\right) }L^{\ast }\right) -\pi \left( 
\mathscr{L}X\right) .
\end{eqnarray*}
This gives the desired result.
\end{proof}

%\noindent $\square$

\subsubsection{Filtering a Cavity Mode}
\label{subsec:cavity_filter}
We wish to derive the filter for a cavity mode $a$ with frequency $\omega $
and damped by a input field (the lossy cavity mode!) to which is undergoing continuous indirect
measurement of the output quadrature. Specifically, the unitary evolution of
thew mode+field satisfies the QSDE 
\begin{eqnarray*}
dU\left( t\right) =\left\{ LdB\left( t\right) -L^{\ast }dB\left( t\right) -(%
\frac{1}{2}L^{\ast }L+iH)dt\right\} U\left( t\right)
\end{eqnarray*}
with initial condition $U\left( 0\right) =1$ and 
\begin{eqnarray*}
H=\omega a^{\ast }a,\quad L=\sqrt{\gamma }a.
\end{eqnarray*}
In this case the Heisenberg dynamics are linear for $a_{t} \triangleq j_{t}\left( a\right)$, and if we begin in
a Gaussian state for system, then we should have a linear Gaussian filtering problem.

\begin{proposition}[Quantum Kalman Filter]
Assuming a Gaussian initial state, the filtered estimate $\hat{a}_t=\pi _{t}\left( a\right)$ satisfies
\begin{eqnarray}
d\widehat{a_{t}}=-\left( \frac{\gamma }{2}+i\omega \right) \widehat{a_{t}}dt+%
\sqrt{\gamma }\left( \mathscr{W}\left( t\right) e^{i\theta \left( t\right) }+\mathscr{V}\left(
t\right) e^{-i\theta \left( t\right) }\right) dI^\theta (t)   \label{eq:filter_a}
\end{eqnarray}
where $\mathscr{W}\left( t\right) $ and $\mathscr{V}\left( t\right) $ are deterministic function
satisfying the coupled equations 
\begin{eqnarray}
\frac{d\mathscr{V}}{dt} &=&-\gamma \mathscr{V}-\gamma \left| \mathscr{V}+e^{2i\theta \left( t\right)
}\mathscr{W}\right| ^{2},  \label{eq:Ricatti_V} \\
\frac{d\mathscr{W}}{dt} &=&-\left( \gamma +2i\omega \right) \mathscr{V}-\gamma \left(
e^{-i\theta \left( t\right) }\mathscr{V}+e^{i\theta \left( t\right) }\mathscr{W}\right) ^{2}
\label{eq:Ricatti_W}
\end{eqnarray}
with $\mathscr{V}\left( 0\right) =\mathbb{E}\left[ a^{\ast }a\right] -\mathbb{E}\left[
a^{\ast }\right] \mathbb{E}\left[ a\right] $ and $\mathscr{W}\left( 0\right) =\mathbb{E%
}\left[ a^{2}\right] -\mathbb{E}\left[ a\right] ^{2}$. and $\mathscr{W}\left( 0\right) =\mathbb{E%
}\left[ a^{2}\right] -\mathbb{E}\left[ a\right] ^{2}$.
\end{proposition}

\begin{proof}
Setting $X=a$ in (\ref{eq:q_filter}) yields (\ref{eq:filter_a}) with 
\begin{eqnarray*}
\mathscr{V}\left( t\right) &=&\pi _{t}\left( a^{\ast }a\right) -\pi _{t}\left( a^{\ast
}\right) \pi _{t}\left( a\right) , \\
\mathscr{W}\left( t\right) &=&\pi _{t}\left( a^{2}\right) -\pi _{t}\left( a\right)
^{2}.
\end{eqnarray*}
We need however to determine the SDEs for $\mathscr{V}$ and $\mathscr{W}$. Setting $X=a^{2}$,
leads to 
\begin{eqnarray*}
d\pi \left( a^{2}\right) &=&-(\gamma +2i\omega )\pi \left( a^{2}\right) dt \\
&&+\sqrt{\gamma }\left\{ \pi \left( a^{3}e^{i\theta }+a^{\ast
}a^{2}e^{-i\theta }\right) -\pi \left( a^{2}\right) \pi \left( ae^{i\theta
}+a^{\ast }e^{-i\theta }\right) \right\} dI^\theta (t) \\
&\equiv &-(\gamma +2i\omega )\pi \left( a^{2}\right) dt+2\sqrt{\gamma }%
\left\{ e^{i\theta }\mathscr{W}+e^{-i\theta }\mathscr{V}\right\} \pi \left( a\right) dI^\theta (t).
\end{eqnarray*}
Which now involves cubic terms. These however may be replaced by lower order moments using the Gaussianity property.
From  the identities  $\pi \left( \left( a-\pi \left(
a\right) \right) ^{3}\right) =0$ and $\pi \left( \left( a-\pi \left(
a\right) \right) ^{\ast }\left( a-\pi \left( a\right) \right) ^{2}\right) =0$, we obtain
\begin{eqnarray*}
\pi (a^{3}) &=&3\pi \left( a^{2}\right) \pi \left( a\right) -2\pi \left(
a\right) ^{3}, \\
\pi \left( a^{\ast }a\right) &=&2\pi \left( a^{\ast }a\right) +\pi \left(
a^{\ast }\right) \pi \left( a^{2}\right) -2\pi \left( a^{\ast }\right) \pi
\left( a\right) ^{2}.
\end{eqnarray*}
In this way we may reduce the order.

From the fact that the innovations  process is a Wiener process, we have 
\begin{eqnarray*}
d\pi \left( a\right) ^{2} &=&2\pi \left( a\right) d\pi \left( a\right)
+\left( d\pi \left( a\right) \right) ^{2} \\
&=&-2\left( \frac{\gamma }{2}+i\omega \right) \pi \left( a\right)
^{2}dt+\gamma \left( e^{-i\theta \left( t\right) }\mathscr{V}+e^{i\theta \left(
t\right) }\mathscr{W}\right) ^{2}dt \\
&&+2\sqrt{\gamma }\left\{ e^{i\theta }\mathscr{W}+e^{-i\theta }\mathscr{V}\right\} \pi \left(
a\right) dI
\end{eqnarray*}
and so 
\begin{eqnarray*}
d\mathscr{W}\left( t\right) &=&d\pi \left( a^{2}\right) -d\pi \left( a\right) ^{2} \\
&=&-\left( \gamma +2i\omega \right) \mathscr{W}dt-\gamma \left( e^{-i\theta \left(
t\right) }\mathscr{V}+e^{i\theta \left( t\right) }\mathscr{W}\right) ^{2}dt
\end{eqnarray*}

Similarly, we have 
\begin{eqnarray*}
d\pi \left( a^{\ast }a\right) &=&-\gamma \pi \left( a^{\ast }a\right) dt\nonumber \\
&&+ \sqrt{\gamma }\bigg\{ \pi \left( a^{\ast }a^{2}e^{i\theta }+a^{\ast
2}ae^{-i\theta }\right) -\pi \left( a^{\ast }a\right) \pi \left( ae^{i\theta
}+a^{\ast }e^{-i\theta }\right) \bigg\} dI^\theta (t) \\
&\equiv &-\gamma \pi \left( a^{\ast }a\right) dt+\nonumber \\
&&2\sqrt{\gamma }\bigg\{
e^{i\theta }[\mathscr{V}\pi \left( a\right) +\mathscr{W}\pi \left( a^{\ast }\right)] +e^{-i\theta
}[\mathscr{V}\pi \left( a^{\ast }\right) +\mathscr{W}^{\ast }\pi \left( a\right) ]\bigg\} dI^\theta (t) ,
\end{eqnarray*}
while 
\begin{eqnarray*}
d\left( \pi \left( a^{\ast }\right) \pi \left( a\right) \right) &=&d\pi
\left( a^{\ast }\right) \,\pi \left( a\right) +\pi \left( a^{\ast }\right)
\,d\pi \left( a\right) +d\pi \left( a^{\ast }\right) \,d\pi \left( a\right)
\\
&=&-\gamma \pi \left( a^{\ast }a\right) dt \\
&&+2\sqrt{\gamma }\left\{ e^{i\theta }[\mathscr{V}\pi \left( a\right) +\mathscr{W}\pi \left(
a^{\ast }\right) ]+e^{-i\theta }[\mathscr{V}\pi \left( a^{\ast }\right) +\mathscr{W}^{\ast }\pi
\left( a\right) ]\right\} dI ^\theta (t) \\
&&+\gamma \left| \mathscr{W}e^{i\theta }+\mathscr{V}e^{-i\theta }\right| ^{2}dt
\end{eqnarray*}
and therefore 
\begin{eqnarray*}
d\mathscr{V}\left( t\right) =d\pi \left( a^{\ast }a\right) -d\left( \pi \left(
a^{\ast }\right) \pi \left( a\right) \right) 
=-\gamma \mathscr{V}dt-\gamma \left| \mathscr{W}e^{i\theta }+\mathscr{V}e^{-i\theta }\right| ^{2}dt.
\end{eqnarray*}
So $\mathscr{V}$ and $\mathscr{W}$ obey deterministic ODEs.
\end{proof}

\bigskip 

We now note that
\begin{eqnarray*}
\mathbb{E}\left[ \left( a_{t}-\widehat{a_{t}}\right) ^{\ast }\left( a_{t}-%
\widehat{a_{t}}\right) \right] =\mathbb{E}\left[ \pi _{t}\left( a^{\ast
}a\right) -\pi _{t}(a^{\ast })\pi _{t}\left( a\right) \right] \equiv
\mathscr{V}\left( t\right)  .
\end{eqnarray*}
If we wish to minimize $\mathbb{E}\left[ \left( a_{t}-\widehat{a_{t}}\right)
^{\ast }\left( a_{t}-\widehat{a_{t}}\right) \right] $ over all quadrature
measurements, then we basically minimize $\mathscr{V}\left( t\right) $ over all choices
of $\theta $. 

\subsection{The Bouten-van Handel Formulation}
Without great loss of generality, let us focus on the single input process example. We wish to compute the filter $\pi_t (\cdot )$ for homodyne
quadrature measurement with $Y(t) = B^{\textit{out}}(t) +B^{\textit{out} \ast} (t)$. We revert to the input picture and set $Z(t) = B(t) + B(t)^\ast$.
Let us introduce $\mathfrak{Z}_{t]} = \text{vN} \{ Z(s) : 0 \leq s \leq t \}$. Then the adapted process $F$ is defined by
\begin{eqnarray}
dF(t) = \big\{ L dZ(t) - (\frac{1}{2} L^\ast L + i H) dt \big\} \, F(t), \qquad F(0) = I,
\label{eq:F}
\end{eqnarray}
has the property that $F(t) \in \mathfrak{F}_{t]}^\prime$ and 
\begin{eqnarray}
\mathbb{E} [j_t (X) ] \equiv \mathbb{E} [ F(t)^\ast X F(t) ] ,
\end{eqnarray}
for every system observable $X$. This is established as Lemma 6.2 in \cite{BvH_ref}. As a corollary, they obtain that
\begin{eqnarray}
\pi_t (X) = \frac{ \sigma_t (X) }{\sigma_t (I) },
\label{eq:QKS}
\end{eqnarray}
where 
\begin{eqnarray}
\sigma_t (X) \triangleq U_t^\ast \, \mathbb{E} \big[ F(t)^\ast X F(t) \mid \mathfrak{Z}_{t]} \big] \, U(t)  .
\label{eq:sigma}
\end{eqnarray}
The equation (\ref{eq:QKS}) may be referred to as the \textit{quantum Kallianpur-Striebel} formula.

It is noteworthy that $\sigma_t (X)$ is the output picture counterpart of the input picture condtional expectation 
$\varsigma_t (X) \triangleq \mathbb{E} \big[ F(t)^\ast X F(t) \mid \mathfrak{Z}_{t]} \big]$,
which involves a conditional expectation with respect to input picture process $Z$. In fact, from (\ref{eq:F})  we readily obtain
\begin{eqnarray}
d \varsigma_t (X) = \varsigma_t ( \mathscr{L} X ) \, dt + \varsigma_t ( XL +L^\ast X ) \, dZ(t) .
\label{eq:varsigma}
\end{eqnarray}
Setting $\varpi_t (X) = \varsigma_t (X) /\varsigma_t  (I)$, it is an easy application of stochastic calculus to show that
\begin{eqnarray}
d \varpi_t (X) &=& \varpi_t ( \mathscr{L} X ) \, dt \nonumber \\
&+& \bigg\{ \varpi_t ( XL +L^\ast X ) -\varpi_t ( X  )\varpi_t ( L +L^\ast  )\bigg\}  \big[ dZ(t) -\varpi_t (L+L^\ast ) dt \big] .
\label{eq:varpi}
\end{eqnarray}
Transferring back to the output picture, we then find that the analogues to (\ref{eq:varsigma}) and (\ref{eq:varpi}) are
\begin{eqnarray}
d \sigma_t (X) = \sigma_t ( \mathscr{L} X ) \, dt + \sigma_t ( XL +L^\ast X ) \, dY(t) .
\end{eqnarray}
and
\begin{eqnarray}
d \pi_t (X) &=& \pi_t ( \mathscr{L} X ) \, dt \nonumber \\
&+& \bigg\{ \pi_t ( XL +L^\ast X ) -\pi_t ( X  )\pi_t ( L +L^\ast  )\bigg\}  \big[ dY(t) -\pi_t (L+L^\ast ) dt \big] ,
\end{eqnarray}
and these are just the general Belavkin-Zakai and Belavkin-Kushner-Stratonovich equations, respectively.

%%%%%%%%%%%%%%%%%%%%%%%%%%%%%%%%%%%%%%%%%%%%%%%%%%%%%%%%%%%%%%%%%%%%%%%%%%%%%%%%%%%%%%%
\section{The Input \& Output Pictures}
\label{sec:IO_pictures}
In this section we develop some ideas of controlled flows and measurement-based feedback. To begin with, in classical control theory
one frequently meets the concept of a controlled dynamical system - that is, a system where the evolution is governed by an external
controlling process $W(t)$. The controlling process may be deterministic, or even stochastic. The situation of interest to engineers is
when the controlling process is allowed to depend on the state (or, more generally, past history) of the system itself. This is an example
of feedback, and classically one may have partial observations $V$ of the state of the system - so one effectively wishes to take the controlling
process $W$ to be some function of the observations $V$. $W$ should depend causally on $Y$, and mathematically one would ask that $W$ is adapted
to the filtration of sigma-algebras generated by $V$.

The quantum mechanical version is more involved due to the subtle and notorious issue of measurement. We have discussed how the Hudson-Parthasarathy
theory leads to the output processes. Here we have, for the lack of a better terminology, a \textit{pre-measurement process} $Z$ which is transformed into
the measured output process $Y$ via $Y(t) =U(t)^\ast Z(t) U(t)$. In a controlled unitary dynamics, we would allow the unitary process $U$ to depend on the 
controlling process $W$ in some causal way - $U \equiv U^{[W]}$. This of course means that the measured output $Y$ somehow is influenced by $W$. Making a feedback loop would mean that
$W$ now becomes dependent of $Y$. Ultimately, this implies that the controlling process $W$ somehow is related to the pre-measurement process $Z$.

To try and cut through a logical Gordian knot, we take the simplifying step and identify $Z$ as the controlling process at the outset. While not the most general
situation (we could have $W$ include other exogenous processes), it does allow us to focus on the feedback loop.

The next point of which we need to take account is that the processes $Z$ and $Y$ are incompatible! In fact, it is best to view them as belonging to two different
pictures which we may call the \emph{input picture} and the \textit{output picture}. 
\begin{figure}[htbp]
	\centering
		\includegraphics[width=0.60\textwidth]{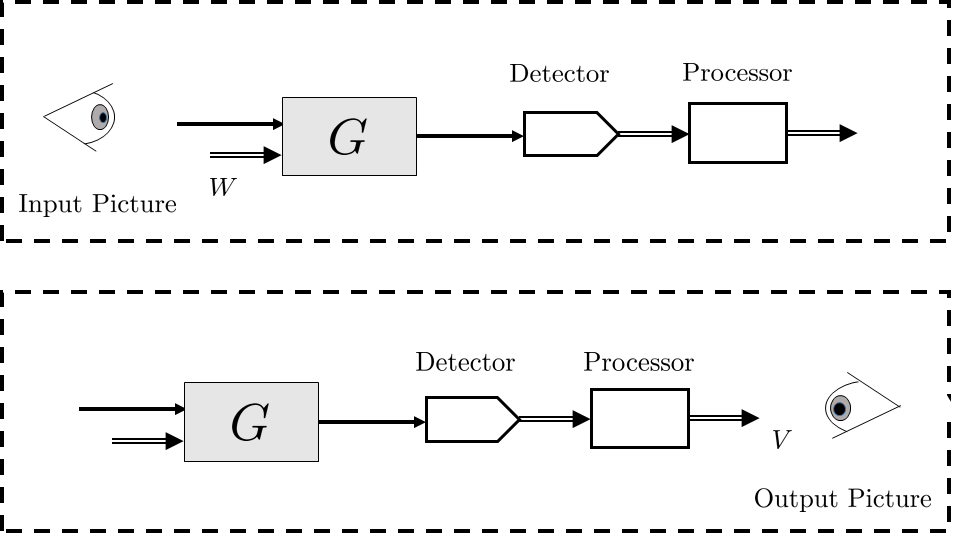}
	\caption{In the input picture, we may view the system as being controlled by some adapted process $W$ which may generally depend on the input: specifically
	$W(t) \in \mathfrak{Z}_{t]}^\prime$. In the output picture, we may be interested in some readout process $V$ which may generally depend on the measurements:
	$V(t) \in \mathfrak{Y}_{t]}^\prime $.}
	\label{fig:wall_eye}
\end{figure}

In the input picture, we describe the world through the controlling process $Z$ - here the unitary $U=U^{[Z]}$
will be influenced by $Z$. More specifically, we have the Hudson-Parthasarathy theory where the $SLH$-coefficients will no longer be fixed system operators,
but generally adapted processes commuting with $Z$. In the output picture, we describe the world through the measured process $Y$. Just as with the Schr\"{o}dinger and Heisenberg pictures,
it is possible to transfer from one to the other, but it is imperative to know which picture you're working in at a particular time, see Figure \ref{fig:wall_eye}.

We now introduce some notation. Let $Z$ be an adapted commutative quantum
stochastic process on the Fock space, and let $\mathfrak{Z}_{t]}=\mathrm{vN}%
\left\{ Z\left( s\right) :0\geq s\leq t\right\} $. 

\begin{definition}
We say that an adapted process, $F$, is causally controlled by $Z$ if, for
each $t$, we have $F\left( t\right) \in \mathfrak{Z}_{t]}^{\prime }$. We often
write $F^{\left[ \left[ Z\right] \right] }$ for emphasis.
\end{definition}

The definition can be expanded by replacing $Z$ with a family $W$ of commutative processes that is enlarged to include exogenous processes.

Note that we will have $F^{\left[ \left[ Z\right] \right] }\left( t\right)
\in \mathfrak{Z}_{t]}^{\prime }\cup \mathfrak{A}_{t]}$ by the adaptedness
requirement.

\begin{definition}
A controlled quantum flow (controlled by $Z$) is the dynamics determined by
the adapted unitary process, $U^{\left[ Z\right] }$, determined by $SLH$%
-coefficients $G^{\left[ \left[ Z\right] \right] }\sim \left( S^{\left[ %
\left[ Z\right] \right] },L^{\left[ \left[ Z\right] \right] },H^{\left[ %
\left[ Z\right] \right] }\right) $ which are adapted processes causally
controlled by $Z$.
\end{definition}

We remark that it is necessary and sufficient that the $SLH$-coefficients satisfy the
previous algebraic conditions - $S^{\left[ \left[ Z\right] \right] }\left(
t\right) $ an $n\times n$ unitary array, and $H^{\left[ \left[ Z\right] %
\right] }\left( t\right) $ self-adjoint - in order that $U^{\left[ Z\right]
}\left( t\right) $ is unitary. For $n=1$, we have
\begin{eqnarray}
dU^{[ Z] }\left( t\right)  &=&\bigg\{ \big( S^{\left[ \left[ Z\right] \right] }( t) -I \big) \otimes d\Lambda \left( t\right) 
+L^{\left[ \left[ Z\right] \right] }\left( t\right) \otimes dB\left( t\right)^{\ast } \nonumber \\
&&-L^{\left[ \left[ Z\right] \right] }\left( t\right) ^{\ast }S^{\left[ \left[ Z\right] \right] }\left( t\right) \otimes dB\left( t\right) \nonumber \\
&&-\big( 
\frac{1}{2}L^{\left[ \left[ Z\right] \right] }\left( t\right) ^{\ast }L^{\left[ \left[ Z\right] \right] }\left( t\right) +iH^{\left[ \left[ Z\right] \right] }\left( t\right) \big) 
\otimes dt \bigg\} U^{\left[ Z\right] }\left(
t\right) ,
\label{eq:U^Z}
\end{eqnarray}
with $U^{\left[ Z\right] }\left( 0\right) =I$. The tensor product in (\ref{eq:U^Z}) is
wrt. $\mathfrak{A}_{t]}\otimes \mathfrak{A}_{[t,\infty )}$. The unitary $U^{\left[ Z%
\right] }\left( t\right) $ would only belong to $\mathfrak{Z}_{t]}^{\prime }$ in
completely trivial situations. Typically it does not commute with
controlling process $Z$. For this reason, we write it as $U^{\left[ Z\right]
}$ so to show that it depends on the controlling process (it does through
its $SLH$-coefficients), rather than as $U^{\left[ \left[ Z\right] \right] }$
since it is not itself a process controlled by $Z$.

The concept of a controlled flow is developed in the input picture. Let us
look at what happens in the output picture. We obtain a second process $%
Y\left( t\right) =U^{\left[ Z\right] }\left( t\right) ^{\ast }Z\left(
t\right) U^{\left[ Z\right] }\left( t\right) $ which corresponds to a
possible measurement process, namely the one obtained from $Z$ using a
component controlled by $Z$. The process $Y$ is then commutative. We may
additionally set $\mathfrak{Y}_{t]}=U^{\left[ Z\right] }\left( t\right) ^{\ast }%
\mathfrak{Z}_{t]}U^{\left[ Z\right] }\left( t\right) ^{\ast }$.

\begin{proposition}
Given an adapted process $F^{\left[ \left[ Z\right] \right] }$ controlled by 
$Z$, the process $\tilde{F}^{\left[ \left[ Y\right] \right] }$ defined by
\begin{eqnarray*}
\tilde{F}^{\left[ \left[ Y\right] \right] }\left( t\right) =U^{\left[ Z%
\right] }\left( t\right) ^{\ast }F^{\left[ \left[ Z\right] \right] }\left(
t\right) U^{\left[ Z\right] }\left( t\right) 
\end{eqnarray*}
is an adapted process controlled by $Y$ (that is, $\tilde{F}^{\left[ \left[ Y%
\right] \right] }\left( t\right) \in \mathfrak{Y}_{t]}^{\prime }\cup \mathfrak{A}%
_{t]}$ for each $t$).
\end{proposition}

\begin{proof}
This is relatively straightforward to see. Let $\tilde{K}\left( t\right) \in 
\mathfrak{Y}_{t]}$, then we should have $\tilde{K}\left( t\right) =U^{\left[ Z%
\right] }\left( t\right) ^{\ast }K\left( t\right) U^{\left[ Z\right] }\left(
t\right) $ for some $K\left( t\right) \in \mathfrak{Z}_{t]}$. Therefore
\begin{eqnarray*}
\left[ \tilde{F}^{\left[ \left[ Y\right] \right] }\left( t\right) ,\tilde{K}%
\left( t\right) \right] =U^{\left[ Z\right] }\left( t\right) ^{\ast }\left[
F^{\left[ \left[ Z\right] \right] }\left( t\right) ,K\left( t\right) \right]
U^{\left[ Z\right] }\left( t\right) \equiv 0
\end{eqnarray*}
since $F^{\left[ \left[ Z\right] \right] }\left( t\right) $ is in the
commutant of $\mathfrak{Z}_{t]}$. Therefore $\tilde{F}^{\left[ \left[ Y\right] %
\right] }\left( t\right) \in \mathfrak{Y}_{t]}^{\prime }$, and it is clearly
adapted.
\end{proof}

The key correspondences between the two pictures are tabulated below:
\begin{eqnarray*}
\begin{tabular}{lll}
& Input Picture & Output Picture \\ \hline
adapted process & $F^{\left[ \left[ Z\right] \right] }\left( t\right) $ & $%
\tilde{F}^{\left[ \left[ Y\right] \right] }\left( t\right) $ \\ 
basic process & $Z$ (control) & $Y$ (measured) \\ 
system observables & $X$ & $j_{t}^{\left[ \left[ Y\right] \right] }\left(
X\right) $ \\ 
noise & $B\left( t\right) $ & $B^{\left[ \left[ Y\right] \right] \text{out}%
}\left( t\right) $ \\ 
filter & $\varpi _{t}^{\left[ \left[ Z\right] \right] }\left( X\right) $ & $%
\pi _{t}^{\left[ \left[ Y\right] \right] }\left( X\right) $ \\ 
filter (Zakai) & $\varsigma _{t}^{\left[ \left[ Z\right] \right] }$ & $%
\sigma _{t}^{\left[ \left[ Y\right] \right] }$ \\ 
conditional state  & $\varrho _{t}^{\left[ \left[ Z\right] \right] }$ & $%
\rho _{t}^{\left[ \left[ Y\right] \right] }$%
\end{tabular}
.
\end{eqnarray*}
To run through these briefly, we have that $j_{t}^{\left[ \left[ Y\right] 
\right] }\left( X\right) $ is defined to be $\tilde{F}^{\left[ \left[ Y%
\right] \right] }\left( t\right) $ when we take $F^{\left[ \left[ Z\right] 
\right] }\left( t\right) =X$. Likewise,  $B^{\left[ \left[ Y\right] \right] 
\text{out}}\left( t\right) $ is defined to be $\tilde{F}^{\left[ \left[ Y%
\right] \right] }\left( t\right) $ when we take $F^{\left[ \left[ Z\right] 
\right] }\left( t\right) =B\left( t\right) $. The filter estimate for system
observable $X$ at time $t$ - more exactly, for $j_{t}^{\left[ \left[ Y\right]
\right] }\left( X\right) $ - based on measurements $\mathfrak{Y}_{t]}$ is $\pi
_{t}^{\left[ \left[ Y\right] \right] }\left( X\right) $. This may be
transferred back to the input picture to give $\varpi _{t}^{\left[ \left[ Z%
\right] \right] }\left( X\right) $:
\begin{eqnarray*}
\pi _{t}^{\left[ \left[ Y\right] \right] }\left( X\right) =U^{\left[ Z\right]
}\left( t\right) ^{\ast }\varpi _{t}^{\left[ \left[ Z\right] \right] }\left(
X\right) U^{\left[ Z\right] }\left( t\right) .
\end{eqnarray*}
We may write
\begin{eqnarray*}
\pi _{t}^{\left[ \left[ Y\right] \right] }\left( X\right) =\text{tr}\left\{
\rho _{t}^{\left[ \left[ Y\right] \right] }X\right\} \text{, or }\varpi
_{t}^{\left[ \left[ Z\right] \right] }\left( X\right) =\text{tr}\left\{
\varrho _{t}^{\left[ \left[ Z\right] \right] }X\right\} .
\end{eqnarray*}
We remark that the conditioned state is logically understood as being
conditioned with respect to measurement process - in other words, it is $%
\rho _{t}^{\left[ \left[ Y\right] \right] }$. The Bouten-van Handel
formulation allows a simple derivation of the filter equations by computing
the input picture $\varrho _{t}^{\left[ \left[ Z\right] \right] }$ and
transferring over to the output picture. Of course, the measurements are
made in the output picture though we can use the circumlocution that we are
measuring the process that corresponds to $Z$ in the input picture. While
serviceable, we have to be careful when jumping between the pictures. In the 
$n=2$ case, we could include a beam-splitter - $S$ non-diagonal - in which
case the first output quadrature $Y_{1}=B_{1}^{\text{out}}+B_{1}^{\text{out*}%
}$ is mathematically a linear combination of $Z_{1}=B_{1}+B_{1}^{\ast }$ and 
$Z_{2}=B_{2}+B_{2}^{\ast }$ even though it is output picture form of just $Z_{1}$.

We also remark that the Heisenberg equations now take the form 
\begin{eqnarray}
dj_{t}^{[[Y]]}\left( X\right) =\sum_{\alpha =0}^{n}\sum_{\beta =0}^{n}%
\mathscr{L}_{\alpha \beta ,t}^{[[Y]]}\left( j_{t}^{[[Y]]}(X)\right) \,d\Lambda
^{\alpha \beta }\left( t\right) 
\label{eq:controlled_Heis}
\end{eqnarray}
where $\mathscr{L}_{\alpha \beta ,t}^{[[Y]]}\left( \cdot \right) $ are the
super-operators obtained by replacing the previously constant $SLH$-coefficients in $%
\mathscr{L}_{\alpha \beta }$ with their counterpart$ \left( \tilde{S}^{[[Y]]}(t), \tilde{L}^{[[Y]]}(t),
\tilde{H}^{[[Y]]}(t) \right)$.

The output fields satisfy
\begin{eqnarray}
dB^{\text{out}}\left( t\right) =\tilde{S}^{\left[ \left[ Y\right] \right]
}\left( t\right) dB\left( t\right) +\tilde{L}^{\left[ \left[ Y\right] \right]
}\left( t\right) dt+\sum_{\alpha =0}^{n}\sum_{\beta =0}^{n}\mathscr{L}%
_{\alpha \beta ,t}^{[[Y]]}\big( B\left( t\right) \big) \,d\Lambda
^{\alpha \beta }\left( t\right) .
\label{eq:B_out}
\end{eqnarray}
The final part here is not present in the constant $SLH$-coefficient case. Despite the complications in (\ref{eq:B_out})
it is worth recalling that $Y(t) \equiv U^{\left[ Z\right]
}\left( t\right) ^{\ast }  Z(t) \, U^{\left[ Z\right] }\left( t\right)$ no matter how complicated $Z$ is in terms of the input field processes.

\begin{definition}
A conditional-state controlled system is one where the $SLH$-coefficients
are functions of the current state $\varrho _{t}^{\left[ \left[ Z\right] 
\right] }$.
\end{definition}

The basic idea is presented in Figure \ref{fig:state_control}. We first describe the open loop model 
in the traditional output picture - here the system is driven by a input noise and a controlling process,
and the output field travels to a detector where a feature is measured (this will be the $Y$ process), then
a filter computes the conditional state $\rho_t^{[[Y]]}$. Below that we consider the equivalent description of the
open loop model in the input picture where everything is written in terms of the controlling process $Z$,
and in this picture the filter computes $\varrho_t^{[[Z]]}$. Finally, we make the feedback loop which
corresponds to taking the controlling process to depend on $\varrho_t^{[[Z]]}$. Note that we may make further modifications such as applying a classical input function $r(t)$ in the return loop.

\begin{figure}[h]
	\centering
		\includegraphics[width=0.750\textwidth]{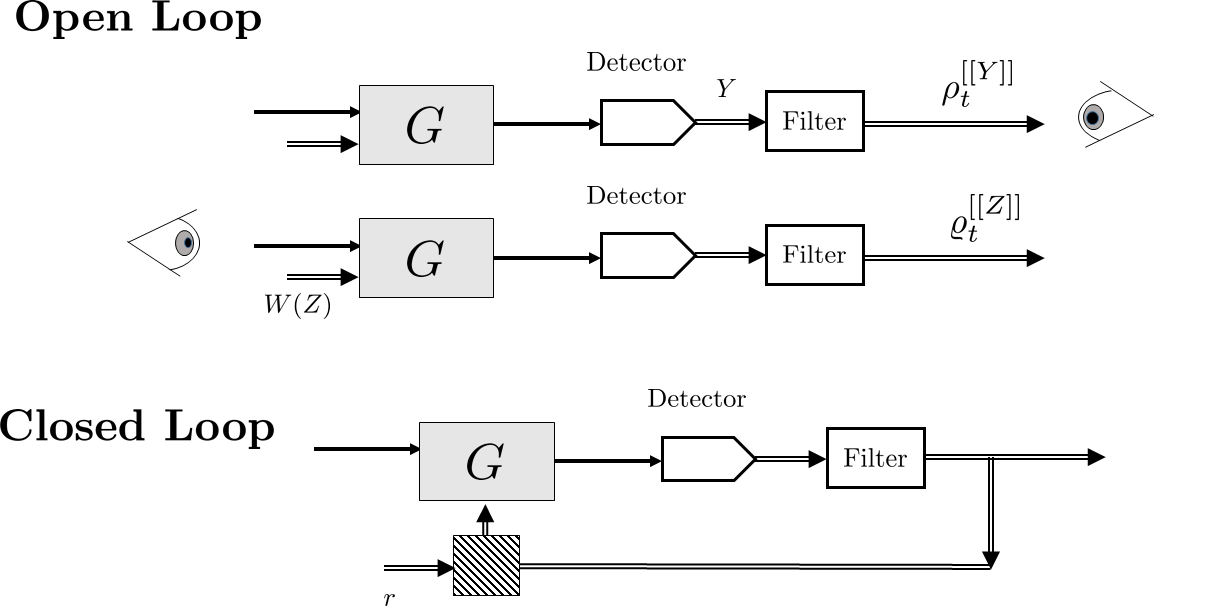}
		\caption{Before feedback, we may present the open loop model in either the input or output picture.
		The closed loop model is depicted underneath.}
	\label{fig:state_control}
\end{figure}

In principle, if we are observing $Y$ then filter will compute $\rho_t^{[[Y]]}$ for us and that this
will be available for feedback. We do not need to use full details of the $\rho^{[[Y]]}_t$, and instead
might be content to use the conditional averages of a fixed set, $\{ X_1 , \cdots , X_n \}$, of system
observables. Returning to the input picture, this would mean that the $SLH$-coefficients are functions of 
the $ \varpi_t^{[[Z]]} (X_k ) = \mathrm{tr} \big\{ \varrho_t^{[[Z]]} X_k \big\}$ (note that these are averages
and so are commuting scalar processes adapted to the pre-measurement process $Z$.

All of this avoids the discussion on how physically the $SLH$-coefficients are made to depend on 
$\varrho_t^{[[Z]]}$, and this ultimately comes down to the ingenuity of the experimentalists
in the laboratory.

\section{Quantum PID Controller}
\label{sec:PID_controller}
In this section we outline a specific model based around the traditional Kalman filter
where we continuously monitor a cavity mode.

As in Subsection \ref{subsec:cavity_filter}, we denote the mode by $a$ but now take $SLH$-coefficients to be
\begin{eqnarray}
G^{[[Z]]}_t \sim \bigg( I, \sqrt{\gamma} a, \omega a^\ast a +i \beta^{[[Z]]} (t) a^\ast -i \beta^{[[Z]] \ast} (t)
a^\ast \bigg),
\label{eq:G_cav_cont}
\end{eqnarray}
where $\beta^{[[Z]]} (t)$ is some as yet unspecified complex-valued adapted process controlled by $Z$.

The Heisenberg-Langevin equation satisfied by $a_t =j_t^{[[Y]]} (a)$ is now
\begin{eqnarray}
da_t = -\left( \frac{\gamma }{2}+i\omega \right) a_{t} \, dt + \tilde{ \beta}^{[[Y]]}(t) \, dt
- \sqrt{\gamma} \, dB(t).
\end{eqnarray}
Note that this is in the output picture and involves  $\tilde{ \beta}^{[[Y]]} (t)\equiv 
U^{[Z]  } (t)^\ast \beta^{[[Z]]} (t) \, U^{[Z]} (t)$ the output picture form of
the control process.

If we wish to measure the $\theta =0$ quadrature, then the corresponding Belavkin-Kalman filter 
for $\hat {a}_t = \pi^{[[Y]]}_t (a)$ will be
\begin{eqnarray}
d\widehat{a}_{t}=-\left( \frac{\gamma }{2}+i\omega \right) \widehat{a}_{t}\, dt
+ \tilde{ \beta}^{[[Y]]}(t) \, dt
+ \sqrt{\gamma }\left( \mathscr{W}\left( t\right) +\mathscr{V}\left(
t\right) \right) dI (t)  , 
\label{eq:filter_a_cont}
\end{eqnarray}
with innovations $dI(t) = dY(t) - \sqrt{\gamma} \big( \hat{a}_t + \hat{a}^\ast_t \big) \, dt$.
Fortunately, the covariances $\mathscr{V}$ and $\mathscr{W}$ will be the same deterministic 
functions encountered before but with $\theta \equiv 0$: that is, they solve
\begin{eqnarray}
\frac{d\mathscr{V}}{dt} &=&-\gamma \mathscr{V}-\gamma \left| \mathscr{V}+\mathscr{W}\right| ^{2},  
\label{eq:Ricatti_V_0} \\
\frac{d\mathscr{W}}{dt} &=&-\left( \gamma +2i\omega \right) \mathscr{V}-\gamma \left(
\mathscr{V}+\mathscr{W}\right) ^{2}
\label{eq:Ricatti_W_0}
\end{eqnarray}
with initial conditions $\mathscr{W}\left( 0\right) =\mathbb{E}\left[ a^{2}\right] -\mathbb{E}\left[ a\right] ^{2}$.

We now revert back to the input picture where (\ref{eq:filter_a_cont}) becomes
\begin{eqnarray}
d\varpi_t^{[[Z]]} (a) =-\left( \frac{\gamma }{2}+i\omega \right)  \varpi_t^{[[Z]]} (a) \, dt
+ \beta^{[[Z]]}(t) \, dt
+ \sqrt{\gamma }\left( \mathscr{W}\left( t\right) +\mathscr{V}\left(
t\right) \right) dJ (t) ,
\label{eq:filter_a_cont_input}
\end{eqnarray}
where $dJ (t) = dZ(t) - \sqrt{\gamma} \big( \varpi_t^{[[Z]]} (a)+\varpi_t^{[[Z]]} (a)^\ast \big) dt$.

\subsection{Proportional Controllers}

In Figure \ref{fig:qfa_loop}, we close the feedback loop and effectively take
\begin{eqnarray}
 \beta_t^{[[Z]]}  = k \left( r(t) -  \varpi_t^{[[Z]]} (a)   \right)    . 
\label{eq:filter_a_feedback}
\end{eqnarray}
Here we have included a reference signal $r(t)$, and have computed the error $e(t) = r(t) -  \varpi_t^{[[Z]]} (a)$.
This is then multiplied by a gain constant $k$ before being fed back in as the controlling process $\beta^{[[Z]]}_t$.

Returning once more to the output picture, we have
\begin{eqnarray}
d\widehat{a}_{t}=-\left( \frac{\gamma }{2}+i\omega \right) \widehat{a}_{t} \, dt
+ k \left( r(t) - \hat{a}_t \right) \, dt
+ \sqrt{\gamma }\left( \mathscr{W}\left( t\right) +\mathscr{V}\left(
t\right) \right) dI (t)  .
\label{eq:filter_a_cont_output}
\end{eqnarray}

\begin{figure}[h]
	\centering
		\includegraphics[width=1.00\textwidth]{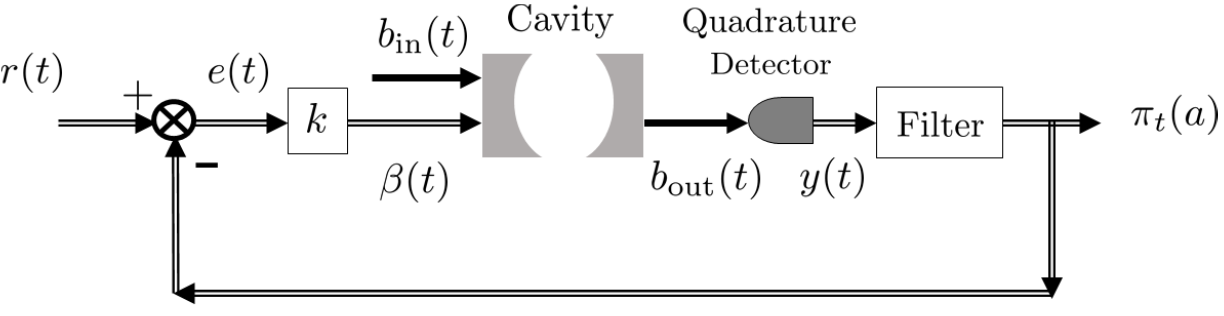}
	\caption{Cavity mode in a feedback loop.}
	\label{fig:qfa_loop}
\end{figure}

In the large gain limit ($k \to \infty$) we may expect the cavity mode to track the reference signal $r(t)$.

The controlled Hamiltonian for the proportional controller model presented here takes the form
$H^{[[Z]]} (t) = \omega a^\ast a+ H_{P}^{[[Z]]}(t) $ where
\begin{eqnarray}
H_{P}^{[[Z]]} (t) =    i k_P \bigg( r(t) - \varpi_t^{[[Z]]} (a) \bigg) + H.c.   
\label{eq:Ham_P}
\end{eqnarray}

\subsection{Proportional-Integral Controllers}
The situation depicted in Figure \ref{fig:qfa_loop} is that of proportional feedback. The controller acts simply by multiplying the error $e(t)$ by a constant gain $k$.
More generally, we may consider a PI (proportional-integral) controller for which the output is
\begin{eqnarray}
\beta (t) = k_P \, e(t) + k_I \int_0^t e(t^\prime) dt^\prime ,
\end{eqnarray}
where $k_P$ and $k_I$ are constants known as the proportional and integral gains respectively. We shall denote the Laplace transform of a function $f=f(t)$ by
$f[s] = \int_0^\infty e^{-st} f(t) dt$ - that is, with square brackets for the transform parameter. We see that $\beta [s] = \mathsf{K}(s) \, e[s]$ where the controller 
transfer function is $\mathsf{K}(s) = k_P + k_I \, \frac{1}{s}$.

It is easy to see that the analogue of (\ref{eq:filter_a_cont_output}) will be
\begin{eqnarray}
\hat{a} [s] = \mathsf{H}(s) \, r[s] + \frac{\mathsf{G}(s)}{1 + \mathsf{G}(s) \mathsf{K}(s)} \,
 \big\{ N(s) - \mathbb{E}_0 [a] \big\} ,
\label{eq:filter_a_PID_out}
\end{eqnarray}
where the open loop transfer function for the system mode is $\mathsf{G}(s) 
= (s+ \frac{\gamma}{2} +i \omega )^{-1}$ and the closed loop transfer function is again $ \mathsf{H}
=  \mathsf{GK}/ (1+\mathsf{GK})$. Note that we have introduce the noise term
$N(s) \triangleq \sqrt{\gamma } \int_0^\infty \left( \mathscr{W}\left( t'\right) +\mathscr{V}\left(
t'\right) \right) dI (t')$  due to the innovations.

For the PI problem, the closed-loop transfer function is
\begin{eqnarray}
\mathsf{H}_{PI} (s) = \frac{sk_P +k_I}{ s^2 + ( k_P +\frac{\gamma}{2} +i \omega) +k_I}.
\label{eq:H_PI}
\end{eqnarray}
and it the special case of a proportional controller ($k_I =0$) we have just
\begin{eqnarray}
\mathsf{H}_{P} (s) = \frac{k_P }{ s + ( k_P +\frac{\gamma}{2} +i \omega) }.
\label{eq:H_P}
\end{eqnarray}

From the perspective of the input picture, we have the controlled dynamics with fixed $S=I$, $L= \sqrt{\gamma} a$ and controlled Hamiltonian $H_t^{[[Z]]} = \omega a^\ast a+ H_{P}^{[[Z]]} (t) + H_{I}^{[[Z]]} (t)$ where
\begin{eqnarray}
H_{I}^{[[Z]]} (t) =   i k_I \int_0^t \big( r(t') - \varpi_{t'}^{[[Z]]} (a) \big) dt' \,  a^\ast + H.c.   
\label{eq:Ham_PI}
\end{eqnarray}

\subsection{PID Controllers}
A widely used controller in classical problem is the PID controller which generalizes the PI controller by allowing and additional derivative term:
\begin{eqnarray}
\beta (t) = k_P \, e(t) + k_I \int_0^t e(t^\prime) dt^\prime  + k_D \dot e (t) .
\end{eqnarray}
At the level of the transfer function this seems at first sight like a simple extension to $\mathsf{K}(s) 
= k_P +  k_I \, \frac{1}{s}+ k_D \, s$. This would correspond to the closed loop transfer function
\begin{eqnarray}
\mathsf{H}_{PID} (s) = \frac{k_D s^2+sk_P +k_I}{ (1+k_D)s^2 + ( k_P +\frac{\gamma}{2} +i \omega) +k_I}.
\label{eq:H_PID}
\end{eqnarray}

However, there is the consistency
problem here that the error is now the difference of the reference signal and the filtered estimate $\varpi_t^{[[Z]]}(a)$ but the latter is a stochastic process
and not differentiable. While we could formally write down a Hamiltonian $H_t^{[[Z]]} $ corresponding to (\ref{eq:Ham_PI}), the new terms would involve the 
formal derivatives of $\varpi_t^{[[Z]]}(a)$ and lie outside the class of model $G^{[[Z]]}_t \sim \big( I, \sqrt{\gamma} a, H^{[[Z]]} (t) \big)$ considered in (\ref{eq:G_cav_cont}).

For convenience, we drop the P and I elements of the controller, and focus
on realizing the D term. Here the feedback law (in the input picture) reads
as
\begin{eqnarray}
\beta ^{\left[\left[ Z\right] \right] }\left( t\right) =k_{D}\dot{e}^{\left[ \left[
Z\right] \right] }\left( t\right) \equiv k_{D}\left( \dot{r}\left( t\right)
-\varpi _{t}^{\left[ \left[ Z\right] \right] }\left( a\right) \right) 
\end{eqnarray}
in which case we have formally
\begin{eqnarray*}
``-iH_{t}^{\left[ \left[ Z\right] \right] }dt"=-i\omega a^{\ast
}a\,dt+k_{D}\left( \dot{r}\left( t\right) a^{\ast }-\text{H.c.}\right)
dt-k_{D}\left( d\varpi _{t}^{\left[ \left[ Z\right] \right] }\left( a\right)
a^{\ast }-\text{H.c.}\right) .
\end{eqnarray*}
We now make the ansatz that the dynamics is based on a specific controlled
flow $G_{t}^{\left[ \left[ Z\right] \right] }\sim \left( I,L_{t}^{\left[ 
\left[ Z\right] \right] },H_{t}^{\left[ \left[ Z\right] \right] }\right) $
whose form is to be deduced. The input picture form of the filter will then
take the form
\begin{eqnarray}
d\varpi _{t}^{\left[ \left[ Z\right] \right] }\left( a\right) =\left\{
\varpi _{t}^{\left[ \left[ Z\right] \right] }\left( \mathscr{L}_{t}^{\left[ 
\left[ Z\right] \right] }a\right) -\lambda _{t}^{\left[ \left[ Z\right] 
\right] }\Xi _{t}^{\left[ \left[ Z\right] \right] }\right\} dt+\Xi _{t}^{%
\left[ \left[ Z\right] \right] }dZ\left( t\right) 
\end{eqnarray}
where $\Xi _{t}^{\left[ \left[ Z\right] \right] }$ is the covariance part of
the filter equation,
\begin{eqnarray}
\Xi _{t}^{\left[ \left[ Z\right] \right] } \triangleq \varpi _{t}^{\left[ \left[ Z%
\right] \right] }\left( aL_{t}^{\left[ \left[ Z\right] \right] }+L_{t}^{%
\left[ \left[ Z\right] \right] \ast }a\right) -\varpi _{t}^{\left[ \left[ Z%
\right] \right] }\left( a\right) \varpi _{t}^{\left[ \left[ Z\right] \right]
}\left( L_{t}^{\left[ \left[ Z\right] \right] }+L_{t}^{\left[ \left[ Z\right]
\right] \ast }\right) ,
\end{eqnarray}
and where $\lambda _{t}^{\left[ \left[ Z\right] \right] }=\varpi _{t}^{\left[ 
\left[ Z\right] \right] }\left( L_{t}^{\left[ \left[ Z\right] \right]
}+L_{t}^{\left[ \left[ Z\right] \right] \ast }\right) $. We therefore have
that
\begin{eqnarray*}
``-iH_{t}^{\left[ \left[ Z\right] \right] }dt" &=&-i\omega a^{\ast
}adt \nonumber \\
&&+k_{D}\left( \dot{r}\left( t\right) a^{\ast }-\left\{ \varpi _{t}^{\left[
\left[ Z\right] \right] }\left( \mathscr{L}_{t}^{\left[ \left[ Z\right] 
\right] }a\right) -\lambda _{t}^{\left[ \left[ Z\right] \right] }\Xi _{t}^{%
\left[ \left[ Z\right] \right] }\right\} a^\ast -\text{H.c.}\right) dt \\
&&-k_{D}\left( \Xi _{t}^{\left[ \left[ Z\right] \right] }a^{\ast }-\Xi _{t}^{%
\left[ \left[ Z\right] \right] \ast }a\right) \left( dB\left( t\right)
+dB\left( t\right) ^{\ast }\right) 
\end{eqnarray*}
where we now use the fact that $dZ=dB+dB^{\ast }$.

The fluctuating parts should now be looked upon as arising from
modifications to the $L$ coefficient. Indeed we should now set
\begin{eqnarray}
L_{t}^{\left[ \left[ Z\right] \right] }=\sqrt{\gamma }a-iF_D^{\left[ \left[
Z\right] \right] } (t)
\label{eq:L_D}
\end{eqnarray}
where $F_{D}^{\left[ \left[ Z\right] \right] } (t)$ is the self-adjoint process
\begin{eqnarray}
F_{D}^{\left[ \left[ Z\right] \right] } (t) \triangleq ik_{D}\left( \Xi _{t}^{%
\left[ \left[ Z\right] \right] \ast }a-\Xi _{t}^{\left[ \left[ Z\right] %
\right] }a^{\ast }\right) .
\label{eq:F_D}
\end{eqnarray}
For the Hamiltonian we shall take $H^{\left[ \left[ Z\right] \right] } (t)
=\omega a^{\ast }a +H_{D}^{\left[ \left[ Z\right] \right] } (t)$ where
\begin{eqnarray}
H_{D}^{\left[ \left[ Z\right] \right] } (t) &=&ik_{D}\left( 
\dot{r}\left( t\right) a^{\ast }-\left\{ \varpi _{t}^{\left[ \left[ Z\right] %
\right] }\left( \mathscr{L}_{t}^{\left[ \left[ Z\right] \right] }a\right)
+\lambda _{t}^{\left[ \left[ Z\right] \right] }\Xi _{t}^{\left[ \left[ Z%
\right] \right] }\right\} a^\ast -\text{H.c.}\right)  \nonumber\\
&&+\frac{1}{2}\sqrt{\gamma }\left( F_{D}^{\left[ \left[ Z\right] \right]
}(t) \, a+a^{\ast }\, F_{D}^{\left[ \left[ Z\right] \right] }(t) \right) .
\label{eq:H_D}
\end{eqnarray}
The final term in (\ref{eq:H_D}) requires some explanation. Ultimately, this goes back
to Wiseman's analysis of direct feedback \cite{Wiseman} and deals with the issue of
feeding back an observed photo-current $y\left( t\right) $ to a plant by
means of an additional Hamiltonian term $F \, y\left( t\right) $ for $F$ a
self-adjoint system operator. This resulted in a shift $L_{0}\rightarrow
L_{0}-iF$ of the coupling operator, $L_{0}=\sqrt{\gamma }a$, but also a
shift of the $H_{0}\rightarrow H_{0}+\frac{1}{2}\left( FL_{0}+L_{0}^{\ast
}F\right) $ of the bare Hamiltonian, $H_{0}=\omega a^{\ast }a$. The
description of this is given in \cite{GouJam09b} in terms of the series product where the
interpretation of the signal being given a second pass through the system (a
system in series with itself!) is derived.

We have now imposed a structural form on the controlled $SLH$-coefficients.
They involve the input-picture filter and so depend on these coefficients in
an intimate manner. However, as the model is still fundamentally linear and
Gaussian, we can deduce explicit expressions. To begin with, the
GKS-Lindbladian generator $\mathscr{L}_{t}^{\left[ \left[ Z\right] \right] }$
is structurally determined by (\ref{eq:L_D}) and (\ref{eq:H_D}), so we deduce that
\begin{eqnarray}
\mathscr{L}_{t}^{\left[ \left[ Z\right] \right] }a &=&
-\left( \frac{1}{2}\gamma +i\omega \right) a
-\sqrt{\gamma }k_{D}\Xi _{t}^{\left[ \left[ Z\right] \right] }\left( a+a^{\ast }\right) \nonumber \\
&&+k_{D}\left( \dot{r}\left( t\right) -\varpi _{t}^{\left[ \left[ Z\right] \right] }
\big( \mathscr{L}_{t}^{\left[ \left[ Z\right] \right] }a\big)
+ \lambda _{t}^{\left[ \left[ Z\right] \right] }\, \Xi _{t}^{\left[ \left[ Z%
\right] \right] }\right) ,
\label{eq:L_a}
\end{eqnarray}
and so, by taking conditional expectations and using $\lambda_t^{[[Z]]} = \sqrt{\gamma}
\big( \varpi_t^{[[Z]]}(a) + _t^{[[Z]]}(a^\ast ) \big)$,
\begin{eqnarray*}
\varpi _{t}^{\left[ \left[ Z\right] \right] }\left( \mathscr{L}_{t}^{\left[ %
\left[ Z\right] \right] }a\right) \equiv -\frac{1}{1+k_{D}}\left( \frac{1}{2}%
\gamma +i\omega \right) \varpi _{t}^{\left[ \left[ Z\right] \right] }(a)+%
\frac{k_{D}}{1+k_{D}}\dot{r}\left( t\right) .
\end{eqnarray*}
Substituting this into the input-picture filter equation yields
\begin{eqnarray*}
d\varpi _{t}^{\left[ \left[ Z\right] \right] }\left( a\right)  =\frac{1}{%
1+k_{D}} \bigg\{ -\left( \frac{1}{2}\gamma +i\omega \right) \varpi _{t}^{\left[ \left[
Z\right] \right] }(a)dt +k_{D}\dot{r}\left( t\right) \bigg\} dt  +\Xi _{t}^{\left[ \left[ Z\right] \right] } \big[ dZ\left( t\right) - \lambda_t^{[[Z]]} \, dt \big] .
\end{eqnarray*}

We also note that the covariance may be computed in terms of the functions $%
\mathscr{V}\left( t\right) $ and $\mathscr{W}\left( t\right) $ encountered
earlier, and indeed
\begin{eqnarray*}
\Xi _{t}^{\left[ \left[ Z\right] \right] } &=&\sqrt{\gamma }\varpi _{t}^{%
\left[ \left[ Z\right] \right] }\left( a^{2}+a^{\ast }a\right) -\sqrt{\gamma 
}\varpi _{t}^{\left[ \left[ Z\right] \right] }\left( a\right) \varpi _{t}^{%
\left[ \left[ Z\right] \right] }\left( a+a^{\ast }\right) -i\varpi _{t}^{%
\left[ \left[ Z\right] \right] }\left( \left[ a,F_{t}^{\left[ \left[ Z\right]
\right] }\right] \right)  \\
&\equiv &\sqrt{\gamma }\left[ \mathscr{V}\left( t\right) +\mathscr{W}\left(
t\right) \right] -k_{D}\Xi _{t}^{\left[ \left[ Z\right] \right] },
\end{eqnarray*}
from which we deduce that
\begin{eqnarray}
\Xi _{t}^{\left[ \left[ Z\right] \right] }=\frac{1}{1+k_{D}}\sqrt{\gamma }%
\left[ \mathscr{V}\left( t\right) +\mathscr{W}\left( t\right) \right] .
\label{eq:Xi}
\end{eqnarray}
Transferring to the output picture then gives
\begin{eqnarray}
d\pi _{t}^{[ [ Y] ] } ( a)  =\frac{1}{1+k_{D}}
\bigg\{ - \big( \frac{1}{2} \gamma +i\omega \big) \pi _{t}^{\left[ \left[ Y\right] \right] }(a)dt
+k_{D} \dot{r}\left( t\right) dt + \sqrt{\gamma }\big[ \mathscr{V}\left( t\right) +%
\mathscr{W}\left( t\right) \big] dI\left( t\right) \bigg\} .
\end{eqnarray}

We note that as $k_D \to \infty$, we obtain $d\pi _{t}^{[ [ Y] ] } ( a) \approx \dot{r}(t) dt$ as
required.

\subsection{The Quantum PID Controller}
We may now collect the various elements together, and summarize our results as follows.

\begin{theorem}
\label{Theorem:PID}
The PID controller for a quantum Kalman filter (with gains $k_P , k_I$ and $ k_D$ respectively) for a cavity mode with natural $SLH$-coefficients $G_0 \sim
(S_0=I , L_0 = \sqrt{\gamma}, H_0 = \omega a^\ast a )$ corresponds in the controlled $SLH$ model given in the input picture by $G_t^{[[Z]]} = (I, L^{[[Z]]} (t) , H^{[[Z]]} (t) )$ with $L^{[[Z]]} (t)$ given by (\ref{eq:L_D}),
\begin{eqnarray}
L ^{\left[ \left[ Z\right] \right] } (t)= L_0
\frac{\sqrt{\gamma} k_{D}}{1+ k_D} \bigg(  
\left[ \mathscr{V}\left( t\right) +\mathscr{W}\left( t\right) \right]
a^{\ast } -\mathrm{H.c.} \bigg) ,
\label{eq:L_D_t}
\end{eqnarray}
and 
\begin{eqnarray}
H^{[[Z]]} (t) = H_0 + H^{[[Z]]}_{P} (t) + H^{[[Z]]}_{I} (t)+ H^{[[Z]]}_{D} (t) ,
\end{eqnarray}
with 
\begin{eqnarray}
H^{[[Z]]}_{P} (t) &=& ik_P \bigg[ r(t) - \varpi_t^{[[Z]]} (a) \bigg] a^\ast + \mathrm{H.c.} ; \nonumber\\
H^{[[Z]]}_{I} (t) &=& ik_I \bigg[ \int_0^t \big( r(t') - \varpi_{t'}^{[[Z]]} (a) \big) dt' \bigg] a^\ast + \mathrm{H.c.} ; \nonumber \\
 H^{[[Z]]}_{D} (t) &=& ik_I \bigg[  \dot{r}(t) - \varpi_{t}^{[[Z]]} (\mathscr{L}_t^{[[Z]]}a) \bigg] a^\ast + \mathrm{H.c.} 
\end{eqnarray}
where (the input picture filter equations)
\begin{eqnarray}
d\varpi_t^{[[Z]]} (a) = \varpi_{t}^{[[Z]]} (\mathscr{L}_t^{[[Z]]}a) \, dt + \frac{\sqrt{\gamma}}{1+k_D}
\big[ \mathscr{V}(t) + \mathscr{W} (t) \big] \, dJ(t),
\label{eq:filter_input_PID}
\end{eqnarray}
and we note the supplementary identity
\begin{eqnarray}
 \varpi_{t}^{[[Z]]} (\mathscr{L}_t^{[[Z]]}a) &=& - \frac{1}{1+k_D} \big( \frac{1}{2} \gamma +i \omega \big) \varpi_t^{[[Z]]} (a) 
\nonumber \\
&& +\frac{k_P}{1+k_D} \big[ r(t) - \varpi_t^{[[Z]]} (a)\big] + \frac{k_I}{1+k_D}\bigg[ \int_0^t \big( r(t') - \varpi_{t'}^{[[Z]]} (a) \big) dt' \bigg] \nonumber \\
&&+ \frac{k_D}{1+k_D} \dot{r} (t).
\label{eq:L_a_PID}
\end{eqnarray}
As before $dJ (t) = dZ(t) - \sqrt{\gamma} \big( \varpi_t^{[[Z]]} (a)+\varpi_t^{[[Z]]} (a)^\ast \big) dt$ is the
input picture version of the innovations, and the covariances $\mathscr{V} (t)$ and $\mathscr{W}(t)$ are the
solutions to (\ref{eq:Ricatti_V_0}) and (\ref{eq:Ricatti_W_0}).
\end{theorem}

\begin{proof}
We note that $L^{[[Z]]} (t)$ corresponds to what we have already seen as the correct coupling in
(\ref{eq:L_D}) to take account of the D term. Here the operator $F^{[[Z]]} (t)$ will be as in (\ref{eq:L_D})
and $\Xi^{[[Z]]} (t) $ will be as in (\ref{eq:Xi}).
Including the P and I Hamiltonians now changes (\ref{eq:L_a}) to
\begin{eqnarray*}
\mathscr{L}_{t}^{\left[ \left[ Z\right] \right] }a &=&
-\left( \frac{1}{2}\gamma +i\omega \right) a
-\sqrt{\gamma }k_{D}\Xi _{t}^{\left[ \left[ Z\right] \right] }\left( a+a^{\ast }\right)  \\
&&+ k_P \big[ r(t) - \varpi_t^{[[Z]]} (a) \big] +k_I \int_0^t \big[ r(t') - \varpi_{t'} (a) \big] dt' \nonumber \\
&&+k_{D}\left( \dot{r}\left( t\right) -\varpi _{t}^{\left[ \left[ Z\right] \right] }
\big( \mathscr{L}_{t}^{\left[ \left[ Z\right] \right] }a\big)
+k_D \lambda _{t}^{\left[ \left[ Z\right] \right] }\Xi _{t}^{\left[ \left[ Z%
\right] \right] }\right) ,
\end{eqnarray*}
and so, by taking conditional expectations, we arrive at (\ref{eq:L_a_PID}). 

Finally the output picture filter equation (that is, the one we really want as external observers) is readily obtained from (\ref{eq:filter_input_PID}) and (\ref{eq:L_a_PID}):
\begin{eqnarray}
d\pi_t^{[[Y]]} (a)  &=& - \frac{1}{1+k_D} \big( \frac{1}{2} \gamma +i \omega \big) \pi_t^{[[Y]]} (a) 
\nonumber \\
&& +\frac{k_P}{1+k_D} \big[ r(t) - \pi_t^{[[Y]]} (a)\big] + \frac{k_I}{1+k_D}\bigg[ \int_0^t \big( r(t') - \pi_{t'}^{[[Z]]} (a) \big) dt' \bigg] + \frac{k_D}{1+k_D} \dot{r} (t) \nonumber \\
&&+ \frac{\sqrt{\gamma}}{1+k_D}
\big[ \mathscr{V}(t) + \mathscr{W} (t) \big] \, dI(t) .
\label{eq:filter_a_hat_PID_out}
\end{eqnarray}
Taking the Laplace transform of (\ref{eq:filter_a_hat_PID_out}), we see that we obtain precisely (\ref{eq:filter_a_PID_out}) - recall the notation $\hat{a} (t) \equiv \pi_t^{[[Y]]} (a)$ - with $\mathsf{K} (s)$ being the PID transfer function $\mathsf{K}_{PID} (s)
= k_P +  k_I \, \frac{1}{s}+ k_D \, s$, and $\mathsf{H}_{PID} (s)$ the closed loop transfer function
appearing in (\ref{eq:H_PID}).
\end{proof}

As the model is described unequivocally in terms of a controlled $SLH$-model, there can be no doubt
that we have a genuine, physically and mathematically well-defined quantum model.

\subsection{Quantum Tuning \& Loop Shaping}
As PID controllers play such a dominant role in practical feedback loop implementations, there are well-developed methods for modeling and tuning them
\cite{Bechhoefer,AH95}.

For instance, it is standard to mitigate against large transient behavior by modifying the feedback law to
\begin{eqnarray}
\beta = k_P ( \mu r - \hat{a} ) +k_I \int_0^t [ r - \hat {a} ] +k_D [ \nu \dot{r} - \frac{d}{dt} \hat{a} ],
\end{eqnarray}
where $\mu$ and $\nu$ are real numbers called \textit{set point weightings}. If we take $\omega \equiv 0$, and drop the D action, then we arrive at the closed-loop transfer function
\begin{eqnarray}
\mathsf{H} (s) = \frac{ \mu k_P s +k_I}{ s^2 +s ( \frac{1}{2} \gamma +i \omega + k_P  ) +k_I},
\end{eqnarray}
and the denominator may be made equal to some desired 2nd order system, say $s^2 +2 \zeta \omega_0 s + \omega^2_0$, by selecting $k_P = 2 \zeta \omega_0 - \frac{1}{2} \gamma$,
and $k_I = \sqrt{\omega_0}$.
This is a basic example of tuning the PID controller. More sophisticated applications, such as loop shaping, exist and there are a variety of techniques (Ziegler-Nichols, Cohen-Coon, etc.,
as well as several PID tuning software products) available. The illustrative model consider in this paper has been a single cavity mode, and this has a simple 1st order transfer function $\mathsf{G}(s)$. However, it is clear that our analysis extends to multi-mode systems since the main requirement for the theory to work is linear gaussian one, and so we expect that the standard techniques of classical PID control may be readily adapted to linear quantum optical systems, for example. This opens up the possibility of extensions to non-linear models, as well as realistic 

\subsection{Feedback to the Detector}
An alternative strategy is adaptive quantum feedback \cite{Wiseman95,AAS} where, instead of a controlled flow, we consider the situation where the feedback is to the detector. This is depicted in Figure \ref{fig:qfa_loop_detector}.
Here we would make the measured quadrature - the $\tilde{\theta}^{[[Y]]} (t)$ parameter - now dependent on the output. This has been treated in \cite{BvH}. We note that despite the different architecture here and the fact that we are controlling the homodyne measurement scheme rather than the system dynamics, the same principle of being able to transfer between the input picture and the output picture applies.

\begin{figure}[h]
	\centering
		\includegraphics[width=0.750\textwidth]{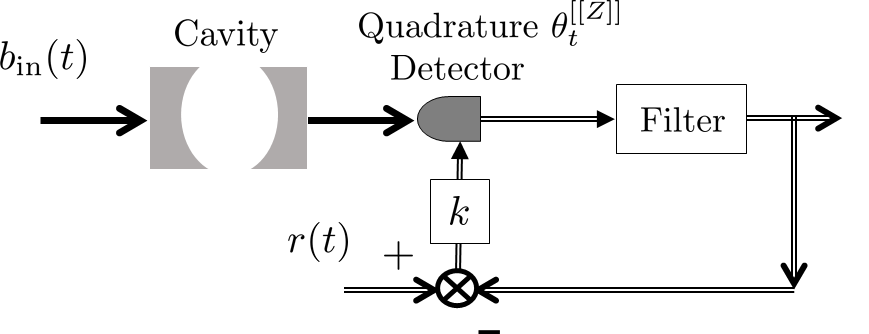}
	\caption{Adaptive control scheme where the quadrature measured, $\theta^{[[Z]]} (t)$ in the input picture, is taken to depend on the filter output.}
	\label{fig:qfa_loop_detector}
\end{figure}

\section{Conclusion}
We have shown that PID controllers can be constructed for quantum systems. At the outset, we had the nagging doubt that the various manipulations of the measured data that can be performed classical by an observer may somehow be incompatible with a consistent quantum model. However, we have shown that this is not the case. At the heart of our discussions is the distinction between the input picture (where we may describe the system as a controlled dynamics and look at the output downline) and the output picture (where we may describe the system in terms of the measured output and filter back upstream). In reality both are necessary, and indeed both are equivalent. To justify the PID controller as a consistent model, we had to set up the feedback law in the output picture, then work back to find the corresponding controlled process description in the input picture. The general PID case is quoted in Theorem \ref{Theorem:PID}, and we may obtain the special cases of P, PI, etc., as simple corollaries. The key point is that the Theorem shows that a controlled $SLH$-dynamical model exists which underpins the feedback model, including the \lq\lq classical" actions done by the controller.
As we have mention, classical PID filters are versatile, well-studied feedback controllers and in principle may be applied to quantum feedback systems. The effect of nonlinearities is an interesting further topic - for instance one might be interested in an anti-wind up scheme - but the $SLH$ theory can be readily extended to nonlinear dynamical models and this will be the subject of further work.

\bibliographystyle{siamplain}

\end{document}